\documentclass[11pt]{article}
\usepackage{amssymb, amsmath, amsthm}
\usepackage{a4wide}
\usepackage{graphicx,xcolor}
\usepackage{prettyref}
\usepackage{algorithm, algpseudocode}

\usepackage[colorlinks=true]{hyperref}
\usepackage[hyperpageref]{backref}
\usepackage[capitalize,nameinlink,noabbrev]{cleveref}
\RequirePackage{xspace}
\hypersetup{
	linkcolor={teal!100!black},
	citecolor={teal!100!black},
	urlcolor={teal!100!black},
    pdftitle={Nearly Tight Lower Bounds for Relaxed Locally Decodable Codes via Robust Daisies},
}

\newtheorem{theorem}{Theorem}[section]

\newtheorem{lemma}[theorem]{Lemma}

\newtheorem{claim}[theorem]{Claim}
\newtheorem{fact}[theorem]{Fact}

\newtheorem{corollary}[theorem]{Corollary}
\newtheorem{definition}[theorem]{Definition}
\newtheorem{remark}[theorem]{Remark}
\newtheorem{observation}[theorem]{Observation}
\newtheorem{thm}{Theorem}
\newtheorem{cor}[thm]{Corollary}
\newtheorem*{prob*}{Problem}

\usepackage{comment}

\def\F{\mathcal{F}}
\def\S{\mathcal{S}}
\def\D{\mathcal{D}}

\def\A{\mathcal{A}}
\def\B{\mathcal{B}}
\def\G{\mathcal{G}}

\newcommand{\supp}{\mathrm{supp}}
\newcommand{\dist}{\mathrm{dist}}

\DeclareMathOperator{\E}{\mathbb{E}}

\newif\iffinal
\finalfalse

\newcommand{\savehyperref}[2]{\texorpdfstring{\hyperref[#1]{#2}}{#2}}

\newrefformat{section}{\savehyperref{#1}{Section~\ref*{#1}}}
\newrefformat{theorem}{\savehyperref{#1}{Theorem~\ref*{#1}}}
\newrefformat{lemma}{\savehyperref{#1}{Lemma~\ref*{#1}}}
\newrefformat{fact}{\savehyperref{#1}{Fact~\ref*{#1}}}
\newrefformat{observation}{\savehyperref{#1}{Observation~\ref*{#1}}}
\newrefformat{def}{\savehyperref{#1}{Definition~\ref*{#1}}}
\newrefformat{claim}{\savehyperref{#1}{Claim~\ref*{#1}}}
\newrefformat{remark}{\savehyperref{#1}{Remark~\ref*{#1}}}
\newrefformat{research question}{\savehyperref{#1}{Research Question~\ref*{#1}}}
\newrefformat{conjecture}{\savehyperref{#1}{Conjecture~\ref*{#1}}}
\newrefformat{construction}{\savehyperref{#1}{Construction~\ref*{#1}}}
\newrefformat{proposition}{\savehyperref{#1}{Proposition~\ref*{#1}}}
\newrefformat{corollary}{\savehyperref{#1}{Corollary~\ref*{#1}}}
\newrefformat{informaltheorem}{\savehyperref{#1}{Informal Theorem~\ref*{#1}}}
\newrefformat{thm}{\savehyperref{#1}{Theorem~\ref*{#1}}}
\newrefformat{algo}{\savehyperref{#1}{Algorithm~\ref*{#1}}}
\newcommand{\pref}{\prettyref}
\newrefformat{obs}{Observation~\ref{#1}}

\title{Nearly Tight Lower Bounds for Relaxed Locally Decodable Codes via Robust Daisies}
\author{Guy Goldberg \thanks{GG is supported by ERC Consolidator Grant 772839 and ISF Grant 2073/21.} \\Weizmann Institute of Science \\ \textsf{guy.goldberg@weizmann.ac.il} \and
Tom Gur\thanks{TG is supported by ERC Starting Grant 101163189 and UKRI Future Leaders Fellowship MR/X023583/1.}\\University of Cambridge\\ \textsf{tom.gur@cl.cam.ac.uk} \and
Sidhant Saraogi\thanks{SS is supported by the NSF CAREER grant CCF-1845125.} \\Georgetown University \\ \textsf{ss4456@georgetown.edu}}
\date{\today}

\begin{document}
\maketitle
\begin{abstract}
    We show a nearly optimal lower bound on the length of linear relaxed locally decodable codes (RLDCs).
    Specifically, we prove that any $q$-query linear RLDC $C\colon \{0,1\}^k \to \{0,1\}^n$ must satisfy $n = k^{1+\Omega(1/q)}$.
    This bound closely matches the known upper bound of $n = k^{1+O(1/q)}$
    by Ben-Sasson, Goldreich, Harsha, Sudan, and Vadhan (STOC 2004).
    
    Our proof introduces the notion of robust daisies, which are relaxed sunflowers with pseudorandom structure, and leverages a new spread lemma to extract dense robust daisies from arbitrary distributions.
    
\end{abstract}

\newpage

\setcounter{page}{1}
\section{Introduction}
    In their influential 2004 paper, Ben-Sasson, Goldreich, Harsha, Sudan, and Vadhan (BGHSV) \cite{BGHSV06} introduced the notion of relaxed locally decodable codes (RLDCs). Similarly to standard locally decodable codes (LDCs), these are error-correcting codes from which individual message bits can be recovered, with high probability, by querying only a few codeword bits, even when the codeword is partially corrupted. However, RLDCs permit a relaxed decoder that, on a small fraction of coordinates, may output the rejection symbol $\bot$ upon detecting corruption.
    
    More precisely, a $(q, \delta, \sigma)$-RLDC $C: \{0, 1\}^k \rightarrow \Sigma^n$ is a code that admits a relaxed decoder $D$ that, given an index $i \in [k]$ and oracle access to $w \in \{0, 1\}^n$ that is $\delta$-close to some codeword $c = C(x)$, satisfies the following conditions.
    \begin{enumerate}
	   \item \emph{Completeness:} if $w = c$, then $D^w(i)=x_i$.
	   \item \emph{Relaxed local decoding:} otherwise, $\Pr[D^w(i)\in\{x_i,\bot\}] \geq \sigma$.
    \end{enumerate}
    As observed in \cite[Lemma 4.10]{BGHSV06}, for $O(1)$-query RLDCs, the two conditions above imply a relaxed decoder that will only output $\bot$ on an arbitrarily small fraction of the message bits.

    This seemingly modest relaxation of LDCs allows for constructions with dramatically better parameters.
    BGHSV constructed $q$-query linear RLDCs with length $n = k^{1+O(1/q)}$.\footnote{While BGHSV only guaranteed that the length is $n = k^{1+O(1/\sqrt{q})}$, Goldreich \cite{Goldreich24} showed that their construction achieves the stronger guarantee, with minor modifications to the analysis.}
    In particular, this implies $O(1)$-query RLDCs with nearly-linear length, while the best known construction of $O(1)$-query (non-relaxed) LDCs has superpolynomial length \cite{yekhanin08,efr12}.
    However, despite the much attention that RLDCs received (cf. \cite{GR18, GG18, GGK19, GRR20, gl19,  AS21, GG21, DGMT22, CGS22, CY22, DGL23, Goldreich24, KM24, CY24}), there are no constructions that improve on BGHSV by achieving length $n = k^{1+o(1/q)}$, and whether such constructions are possible remained an open problem.

    \subsection{Main result}
    We prove a lower bound for linear RLDCs, which closely matches the $n = k^{1+O(1/q)}$ upper bound of BGHSV.
    We do this by first proving a lower bound for RLDCs with \emph{non-adaptive} decoders, and then apply a known reduction by Goldberg \cite{Goldberg24} to deduce the same bound for linear RLDCs.
    
    \begin{thm} \label{thm:main}
        Let $C :\{0,1\}^k\rightarrow \Sigma^n$ be a non-adaptive $(q, \delta, \sigma)$-RLDC, where $q \in \mathbb{N}$, $\sigma > 0$, and $\delta > n^{-\frac{\sigma}{2q}}$.
        Then, 
        \[ 
            n \ge \left( \frac{\sigma^2 \cdot k}{38 q^4  \log^2(|\Sigma|) \cdot \log^2 k } \right)^{1 + \frac{1}{\lceil q/\sigma \rceil}} \;.
        \] 
    \end{thm}

    For simplicity, throughout the rest of this section, we restrict our attention to the standard setting of $\sigma=2/3$ and $\delta=\Omega(1)$, and assume a binary alphabet $\Sigma =  \{0, 1\}$.
    In this setting, and for a constant $q$, \pref{thm:main}, yields the following.
    \begin{cor} \label{corollary:linear_rldc}
        For any linear $q$-RLDC $C:\{0, 1\}^k \to \{0, 1\}^n$ it holds that $n = k^{1+\Omega(1/q)}$.
    \end{cor}
   Our result improves upon the previous state-of-the-art lower bound of $n = k^{1 + \Omega(1/q^2)}$ for linear RLDCs.
        This prior bound was achieved by applying Goldberg's reduction \cite{Goldberg24} to the $n = k^{1 + \Omega(1/q^2)}$ non-adaptive bound in \cite{Goldreich23} (which improved on the $n = k^{1 + \Omega(1/q^2 \log^2 q)}$ bound in \cite{gl19}).
    
    For the constant rate regime, rearranging the terms in \pref{thm:main} yields the following.
    \begin{cor} \label{corollary:const-rate}
        For any linear $q$-RLDC $C:\{0, 1\}^k \to \{0, 1\}^n$, where $n=O(k)$, it holds that $q = \Omega\left(\frac{\log k}{\log \log k}\right)$.
    \end{cor}

    Since linear locally correctable codes (RLCCs) imply RLDCs with the same parameters, both corollaries immediately extend to them.
    
    \paragraph{Beyond non-adaptive RLDCs.}     
    We remark that the machinery developed in \cite{DGL23} provided the means to extend the $n = k^{1 + \Omega(1/q^2 \log^2 q)}$ lower bound in \cite{gl19} to (not necessarily linear) adaptive RLDCs.
    We believe that an adaptation of this machinery will extend our lower bounds to general RLDCs. We leave this to future work.

\subsection{Robust daisies} \label{section:daisies-intro}
    The proof of our main result, \pref{thm:main}, introduces the combinatorial notion of \emph{robust daisies}, which may be of independent interest.
    A distribution over sets is a robust daisy with a kernel $K$ if, after removing the kernel, its support forms a \emph{satisfying set system}.
    That is, a binomial sampling of the universe elements contains a full set from this system (ignoring the elements in $K$) with high probability (see \cref{def:satisfying}).
    Furthermore, this property must hold for any subset of the support, where the required success probability scales exponentially in the subset's density.

    \begin{definition} (Robust daisy) \label{def:robust_daisy-intro}
    A distribution $\mu$ over $\mathcal{P}(U)$ is a \emph{$(p, \varepsilon)$-robust daisy} with kernel $K \subseteq U$, if, for every $\D \subseteq \supp(\mu)$:
    \begin{equation*}
        \Pr_{W \sim \text{Bin}(U, p)}[\exists \; S \in \D, S \subseteq K \cup W] \geq 1-\varepsilon^{\mu(\D)} \: .
    \end{equation*}
    \end{definition}

    For any set $S \in \mathrm{\supp}(\mu)$, we call $S \setminus K$ a \emph{petal} of the robust daisy.
    Put differently, the robust daisy condition dictates that every subset of petals $\{ S \setminus K \mid S \in \D\}$ is $\left(p, \varepsilon^{\mu(\D)}\right)$-satisfying.

    Robust daisies are closely related to robust sunflowers \cite{rossman14}.
    They both require the petals to be a satisfying set system.
    However, the notions differ in two ways: (1) for robust daisies, the kernel is allowed to have an arbitrary structure, rather than being restricted to the intersection of all sets; (2) the robust daisy is a distribution over sets, rather than an unweighted set system, and the satisfying set system condition must hold not only for the support of the distribution, but also for its subsets.

    We remark that, unlike the notion of (non-robust) daisies \cite{gl19,DGL23}, where outside the relaxed kernel each point is only required to be covered by a bounded number of sets, robust daises capture a pseudorandom structure, known as \emph{spreadness}, outside of the kernel.
    Indeed, to argue about robust daisies, we prove a new spread lemma that is applicable to spread families of small sets (see \pref{section:spread}).

    \paragraph{Robust daisy lemma.} 
    Our main structural result concerning robust daisies shows that it is always possible to extract a dense robust daisy from a distribution over small sets.  

    \begin{lemma} (The Robust Daisy Lemma; informally stated, see \pref{lemma:robust_daisy}) \label{lemma:robust_daisy-intro}
        Fix $q\in\mathbb{N}$ and a set $U$ of size $n$.
        Let $\mu$ be a distribution over $\mathcal{P}(U)$ such that $|S| \leq q$ for every $S \in \supp(\mu)$.
        Then, there exists $\D \subseteq \supp(\mu)$ with $\mu(\D) \ge 0.99$ and a kernel $K \subseteq U$ with $|K| = o(n)$ such that the conditional distribution $\mu_{\D}(x)=\frac{\mu(x)}{\mu(\D)}$ is a $(p, \varepsilon)$-robust daisy with kernel $K$, where
        \[
            p = n^{-\Theta(1/q)} \quad \text{and} \quad \varepsilon = 2^{-\Omega(|K|)} \: .
        \]
    \end{lemma}

    The conceptual message of the robust daisy lemma stands in sharp contrast to that of the robust sunflower lemma.
    While the robust sunflower lemma shows the existence of a small, highly structured subset within the set system, the robust daisy lemma instead extracts a pseudorandom approximation of the entire distribution over the set system. 
    To use a metaphor, if a robust sunflower is a small precious ``gem'' that can be found within any large enough mountain, a robust daisy is (approximately) \emph{the entire mountain}.

\subsection{Related work}
    LDCs, LCCs and their relaxed counterparts have attracted significant attention in recent years.
    See the works of Yekhanin \cite{Yekhanin12} and Kopparty and Saraf \cite{KS17} and references therein for comprehensive surveys of LDCs, LCCs and their applications.

    \paragraph{RLDCs constructions.} 
    The constructions of RLDCs and RLCCs can be separated into two main parameter regimes: constant query complexity, and constant rate.

    In the constant rate regime, the state-of-the-art code is the construction by Cohen and Yankovitz \cite{CY24}.
    They construct a linear RLCC with rate arbitrarily close to $1$, and query complexity $q = (\log{n})^{2+o(1)}$.
    This construction builds upon the result by Kumar and Mon \cite{KM24}, which shows a similar code but with query complexity $q = (\log{n})^{O(1)}$.

    In the constant query regime, the original work of \cite{BGHSV06} claimed to achieve RLDC with constant query complexity $O(q)$ and length $n = O(k^{1 + 1/\sqrt{q}})$.
    In fact, \cite{Goldreich24} showed that this construction actually achieves $n = O(k^{1 + 1/q})$, which still makes it the current state-of-the art RLDC construction with constant query complexity.

    The work of \cite{GRR20} introduced the notion of RLCCs, constructing such a code with constant query complexity, but with a worse length tradeoff.
    Chiesa, Gur, and Shinkar \cite{CGS22} constructed an improved RLCC, achieving length $n = O(k^{1 + 1/\sqrt{q}})$ (matching the original BGHSV claim).
    This was later improved by Asadi and Shinkar \cite{AS21}, who constructed an RLCC with length $n = O(k^{1+1/q})$, matching the actual (and stronger) bound of the BGHSV construction.

    \paragraph{Lower bounds.}
    In recent decades, extensive research has been conducted on lower bounds for (non-relaxed) LDCs in various regimes \cite{KT00, KW03, Woodruff07, Woodruff12, AGKM22, JM25, BHKL25}.

    Gur and Lachish \cite{gl19} presented the first lower bound for relaxed LDCs.
    Specifically, they showed that any non-adaptive RLDC requires a block length of $n = k^{1+\Omega \left(\frac{1}{q^2 \log^2 q} \right)}$.
    For the adaptive case, they established a lower bound of $n = k^{1+\Omega \left(\frac{1}{2^{2q} \log^2 q} \right)}$.

    The result of \cite{gl19} was extended to additional settings, such as proofs of proximity and property testing, and to the adaptive setting by Dall’Agnol, Gur and Lachish \cite{DGL23}.
    Specifically, they extended the lower bound of $n = k ^ {1+\Omega \left(\frac{1}{q^2 \log^2 q} \right)}$ to \emph{adaptive} RLDCs.

    Goldreich \cite{Goldreich23} surveyed and simplified the work of \cite{gl19}, without employing the newer techniques of \cite{DGL23}.
    He established an improved bound of $n = k ^ {1+\Omega(1/q^2)}$ for the non-adaptive case, 
    and a bound of $n = k ^ {1+\Omega(1/q^3)}$ for the adaptive case (which is weaker than the one presented in \cite{DGL23}).

    Goldberg \cite{Goldberg24} presented a generic reduction that transforms any lower bound for non-adaptive RLDCs and extends it to (possibly adaptive) linear RLDCs.
    Applying this reduction to the bound from \cite{Goldreich23} extends the $n = k ^ {1+\Omega(1/q^2)}$ lower bound to all  linear RLDCs.

    \paragraph{Spreadness.}    
    Our techniques draw on the powerful connection between spreadness and robust combinatorial structures, a link that has been central to recent breakthroughs.
    
    The concept of spreadness for distributions was introduced by Talagrand \cite{talagrand10}.
    A version of the spread lemma, with roots in Rossman \cite{rossman14}, was famously used by Alweiss, Lovett, Wu, and Zhang \cite{ALWZ21} (building on \cite{LSZ20}) to prove that any sufficiently spread set system contains a robust sunflower. This directly led to a breakthrough on the sunflower lemma.
    This line of work, and the spread lemma itself, has since been significantly strengthened \cite{FKN21, PP23, Rao2020} and has found numerous applications across combinatorics and computer science (among others, see \cite{ALWZ21, FKN21, PP23, CKR22, CGRSS25}).
    For a detailed survey, we refer the reader to \cite{rao2025story}.

\subsection{Open problem} 
    Our work leaves several interesting directions for further research.
    We highlight one open question that we find particularly compelling.
    
    In the constant rate regime, where $n = O(k)$, \pref{thm:main} implies a lower bound on the query complexity, of $q = \Omega(\frac{\log k}{\log \log k})$.
    On the other hand, the recent state-of-the-art construction of a constant-rate RLDC by Cohen and Yankovitz \cite{CY24}, achieves $q = O(\log^2k)$.
    An important open problem is to close the quadratic gap that still remains in this regime.
\section{Proof overview} \label{section:proof-overview}
    The proof of the RLDC lower bound in \pref{thm:main} consists of the following three high-level steps.

\begin{itemize}
    \item \textbf{Step 1: Reduction to a combinatorial problem.}
    First, we reduce the problem of proving a lower bound for RLDCs to a purely combinatorial problem: finding a specific structure, a dense robust daisy, within the decoder's query-set distribution.
    
    \item \textbf{Step 2: From spreadness to robust daisies.}
    Second, we establish the key link between spreadness and our new notion of robust daisies.
    We introduce a set-theoretic property, which is a generalization of the well-known notion of \emph{set spreadness}.
    We prove a new \emph{small-set spread lemma} which shows that any spread set system is \emph{satisfying}, which is the required structure of the robust daisy outside the kernel.

    \item \textbf{Step 3: finding spreadness.} Third, we prove a \emph{spreadness extraction lemma}.
    We show that every distribution over sets can be made spread, by \emph{puncturing} (removing $o(n)$ elements from the universe), and \emph{conditioning} (restricting the distribution to a large-measure subset of its support).
\end{itemize}

    These three components chain together to prove the main theorem.
    We apply the spreadness extraction lemma (Step 3) to the RLDC's query-set distribution to find a large, spread substructure.
    Our small-set spread lemma (Step 2) then proves this structure is a robust daisy.
    Finally, by our reduction (Step 1), the existence of this robust daisy within the decoder's queries implies the $n = k^{1+\Omega(1/q)}$ lower bound. We proceed to elaborate on each of these three steps.

\subsection{RLDC lower bounds via robust daisies}
    Our first contribution is conceptual: we abstract and generalize the argument underlying the RLDC lower bound in \cite{gl19}.
    This abstraction is crucial, as it reveals the bottleneck common to all previous lower bounds, including \cite{DGL23,Goldreich23,Goldberg24}, and provides the generality that is needed to surpass the barrier of $n = k^{1+\Omega(1/q^2)}$ lower bounds.

    In the following, let $C:\{0, 1\}^k \to \{0, 1\}^n$ be a non-adaptive $q$-RLDC; that is, for each decoding index $i \in [k]$ the decoder's queries are determined by a query-set distribution $\mu_i$ over $q$-tuples of codeword coordinates.

    We show that if the relaxed decoder is structured in the sense that each distribution $\mu_i$ constitutes a robust daisy (\pref{def:robust_daisy-intro}), then the following lower bound on the code's block length must hold.

    \begin{lemma} (informally stated, see \pref{lemma:daisy-to-lowerbound}) \label{lemma:daisy-to-lowerbound-inf}
        Let $C :\{0,1\}^k\rightarrow \{0,1\}^n$ be a non-adaptive $q$-query RLDC.
        If for every $i \in [k]$, the query-set distribution $\mu_i$ is a $(p, \varepsilon_i)$-robust daisy with a kernel $K_i \subseteq [n]$ such that $|K_i| = o(n)$ and $\varepsilon_i = 2^{-\Omega(|K_i|)}$, then $k \leq pn$.
    \end{lemma}

    We defer a detailed overview of the proof of \pref{lemma:daisy-to-lowerbound-inf} to \pref{section:lb-overview}, which we encourage readers unfamiliar with the techniques of \cite{gl19} to review first.
    Our focus here is on how the reduction to a combinatorial problem, which \pref{lemma:daisy-to-lowerbound-inf} provides, allows us to overcome the limitations of previous approaches.

    Towards this end, note that the query-set distribution $\mu$ of a general RLDC might not form a robust daisy, but rather an arbitrary set of distributions $\{\mu_i\}_{i\in[k]}$ supported on $q$-sets.
    However, to apply the reduction, it suffices that each $\mu_i$ can be \emph{approximated} by a robust daisy.

    To see this, fix $i\in[k]$ and denote $\mu = \mu_i$.
    Note that if there exists a dense sub-family $\D \subseteq \supp(\mu)$ such that the conditional distribution $\mu_{\D}$ is a robust daisy, we can slightly modify the operation of the relaxed decoder: 
    instead of sampling a set $S \sim \mu$ (as the original decoder does), the modified decoder will sample a set $S \sim \mu_{\D}$.
    Hence, if $\D$ is dense enough (say, with $\mu(\D) \geq 0.99$), then the soundness probability of the modified decoder is only slightly worse than that of the original one.

    The above discussion, combined with \pref{lemma:daisy-to-lowerbound-inf}, reduces the task of proving RLDC lower bounds to the following, purely combinatorial problem.

    \begin{prob*}[Robust daisy extraction]
        Given an arbitrary distribution $\mu$ supported on $q$-sets, extract a dense $(p, \varepsilon)$-robust daisy with kernel $K$ such that $\varepsilon = 2^{-\Omega(|K|)}$, while minimizing $p$.
    \end{prob*}    
    
    Our reduction to the problem of extracting robust daisies, which are closely related to robust sunflowers (see \pref{section:daisies-intro}), exposes the connection to the pseudorandom structure captured by \emph{spread lemmas}, which lies at the heart of robust sunflower lemmas and is also a key component in our proof.
    
    We remark that previous RLDC lower bounds did not employ an abstract reduction; rather, they directly analyzed the query-set distribution to identify specific structures.
    Recasting the works of \cite{gl19,DGL23,Goldreich23} through the lens of robust daisy extraction, the structures they isolate are in fact $(p,\varepsilon)$-robust daisies, albeit with relatively weak parameters: to achieve $\varepsilon$ that is exponentially small in the kernel's size, they need $p=n^{-\Theta(1/q^2)}$.
    The quadratic dependency in the query complexity gives the corresponding factor in the lower bound.\footnote{
    The robust daisy extracted in \cite{gl19} is of density of only $\mu(\D) \geq 1/q$.
    Therefore, they had to employ a soundness amplification process, which increases the query complexity to $q \log q$.
    This increment gives the term $q^2 \log ^2 q$ in their lower bound.
    Using an improved construction, \cite{Goldreich23} is able to extract a robust daisy with density arbitrarily close to $1$ (but still with $p=n^{-\Theta(1/q^2)}$), which avoids the need for soundness amplification, yielding the improved lower bound.}
    Moreover, as we shall see next, there are explicit counterexamples showing that these structures cannot attain the parameters needed to surpass $n = k^{1+\Omega(1/q^2)}$ lower bounds.

\subsection{Extracting a robust daisy}
    In light of the reduction above, we can set aside RLDCs and focus on the combinatorial problem of robust daisy extraction. 
    For simplicity, let us denote a $(p, \varepsilon)$-robust daisy with $\varepsilon = 2^{-\Omega(|K|)}$ as a \emph{$p$-robust daisy}.

    We begin by examining the methods used in previous lower bounds to extract structures that can be viewed as $p$-robust daisies, and explain a barrier for such approaches. 

    \paragraph{The bottleneck: $t$-daisies.}
    The combinatorial structure extracted in all previous works \cite{gl19,DGL23,Goldreich23} is that of a \emph{$t$-daisy}.\footnote{Technically, the argument in \cite{Goldreich23} avoids the notion of $t$-daisies, and the query sets are simply divided into heavy, medium, and light elements. However, the essence of the structure is maintained, as the heavy elements play the role of the kernel, and the light elements admit the bounded intersection property.} 
    A set system is a $t$-daisy if, outside a kernel $K$, each element is contained in at most $t$ sets.

    Reframing \cite{gl19} through the abstraction of robust daises, their argument regarding $t$-daisies can be seen as showing for any set system containing $O(n)$ sets, the following holds:
    \begin{enumerate}
        \item It is always possible to extract a $t$-daisy with $t |K| = O(n^{1-1/q})$.\footnote{Notice that extracting a $t$-daisy with $t |K| = O(n)$ is trivial, since for a constant $q$, there are at most $O(1/t)$ elements with degree larger than $t$.}
        \item If a $t$-daisy satisfies the condition above, then it is a $n^{-\Theta(1/q^2)}$-robust daisy.
    \end{enumerate}
    We remark that the second item is shown by simple, first-principle arguments.
    Namely, using a greedy process, \cite{gl19} finds a family of disjoint sets in the $t$-daisy.
    Then, they directly calculate the probability to sample one of these sets.
    
    We argue that $p=n^{-\Theta(1/q^2)}$ is the best (i.e., minimal) sampling probability that can be achieved by extracting $t$-daisies; hence, they cannot imply stronger RLDC lower bounds.

    First, there exist $t$-daisies satisfying the condition $t |K| = O(n^{1-1/q})$ which are \emph{not} $p$-robust daisies for any $p=n^{-o(1/q^2)}$.
    To see this, fix $t=n^{1-1/q}$, and consider the set system consisting of $n/t=n^{1/q}$ sets of size $q$, each repeating $t$ times.
    This set system is a $t$-daisy with an empty kernel, and satisfies $t |K| = O(n^{1-1/q})$. However, by a straightforward calculation, one needs $p=n^{-O(1/q^2)}$ to sample a full set with a constant probability.\footnote{For a counterexample that avoids repetitions of the same set, one can add a unique vertex to each of the $n$ sets, while maintaining the sampling bound.}

    Nevertheless, one might wonder: is it possible to extract a $t$-daisy with a better relation between $t$ and $K$?
    This might imply that this $t$-daisy is a $p$-robust daisy with $p=n^{-o(1/q^2)}$. Alas, this is impossible.
    The process for extracting $t$-daisies in \cite{gl19} is optimal; one cannot find better $t$-daisies in arbitrary set systems.

    To illustrate this, we present a $(q+1)$-uniform set system, with $n$ sets over $O(n)$ base elements.
    This set system has the property that for every value of $t$, and setting the kernel $K$ to contain the elements with degree larger than $t$, it holds that $t |K| = \Omega(n^{1-1/q})$.
    Consider the $k$-regular tree with $q+1$ levels, with $k=n^{1/q}$.
    Each level $\ell$ in the tree has $k^{\ell - 1}$ vertices (the root is on level $1$), and the last level where $\ell = q + 1$ has $n$ leaves.
    The $n$ sets in the system are the unique $n$ root-to-leaf paths in the tree.
    The degree of each element in level $\ell$ is $k^{q + 1 - \ell}$, which is the number of root-to-leaf paths going through it.
    The degrees decrease from $k^q = n$ at the root to $k^0 = 1$ at the leaves.
    
    Now, consider any degree threshold $t$.
    This threshold must fall between the degrees of two adjacent levels.
    Namely, $t \in [k^{q - \ell}, k^{q + 1 - \ell})$ for some $\ell$ between $0$ and $q$.
    The kernel $K$ of a $t$-daisy now must include all vertices in the top $\ell$ levels --  otherwise there is an element included in $k^{q + 1 - \ell} > t$ or more petals.
    The size of this kernel is therefore at least the total number of vertices in these levels:
    \[
        |K| \geq \sum_{j=1}^{\ell} k^{j-1} = \Theta(k^{\ell-1}) \; .
    \]
    On the other hand, $t \geq k^{q - \ell}$, and together this implies:
    \[
        t|K| = \Omega(k^{q - \ell} \cdot k^{\ell-1}) = \Omega(k^{q - 1}) = \Omega(n^{1-1/q}) \; .
    \]

    To summarize: the notion of $t$-daisies is, on the one hand, not strong enough to imply the desired $n^{-\Theta(1/q)}$ sampling probability, and on the other hand, too strong, so improving the extraction process is impossible.
    
    Nevertheless, as we shall see next, it is possible to extract a $n^{-\Theta(1/q)}$ robust daisy from every distribution supported over sets of size at most $q$, even from the tree set system above -- this is guaranteed by our Robust Daisy Lemma (\pref{lemma:robust_daisy-intro}).
    For that, however, we need new ideas that avoid the bottleneck of $t$-daisies.

    \paragraph{$(m,k)$-spreadness.} 
    In this work, we introduce the notion of \emph{$(m,k)$-spreadness}, which is a generalization and strengthening of the known notion of $k$-spreadness \cite{talagrand10, MNSZ25}.
    This stronger notion is essential for extracting robust daisies that outperform those from previous approaches.

    In the following, we use $\langle T \rangle$ to denote the family of all subsets of $U$ that contain $T$, and then $\mu(\langle T \rangle)$ is the total density of all subsets of $U$ that contain $T$.

    \begin{definition}[$(m, k)$-spread distributions] \label{def:spread-distribution-intro} 
        Let $\mu$ be a distribution over $\mathcal{P}(U)$, let $k > 1$ and let $m \in (0, 1]$.
        We say that $\mu$ is \emph{$(m, k)$-spread} if for any non-empty set $T \subseteq U$,
        \[
           \mu(\langle T \rangle) \leq \frac{m}{k^{|T|}} \:.
        \]
    \end{definition}

    This new notion coincides with the standard $k$-spreadness for $m=1$.
    For $m<1$, however, it is stronger.
    We believe this generalization, and the connection it creates between RLDCs to the literature on spreadness, might be of independent interest.
    
    We remark that even with $m = 1$, the spreadness condition is strictly stronger than the condition required from the petals of a $t$-daisy.
    To see that, fix $m = 1$, and
    assume the uniform distribution over a set system $\F$ is $k$-spread.
    Then, applying the spreadness condition to $T = \{ x \}$ for every universe element $x$, we deduce that every element is contained in at most $\frac{|\F|}{k}$ sets.
    That is, $\F$ is a $\frac{|\F|}{k}$-daisy with an empty kernel.
    In other words, the bounded intersection requirement is similar to asking for spreadness, but for singletons only.
    Spreadness is a much stronger requirement, as it applies to any subset of $U$.

    Our generalized notion of $(m,k)$-spreadness allows us to prove a new spread lemma, which is useful specifically when the support of the distribution is over small sets.
    Before stating the new lemma, let us demonstrate why we need the new generalized definition.

    \paragraph{$(m,k)$-spreadness vs $k$-spreadness.}
    The main benefit of the new notion comes from the following observation:
    One cannot hope for spreadness parameter better than $k=n^{\Theta(1/q)}$.
    Consider the uniform distribution over $n$ distinct sets, each of size $q$.
    Let $T$ be one of these $n$ sets.
    The probability of picking $T$ is $\mu(T) = 1/n$.
    Since $T \in \langle T \rangle$, the total density $\mu(\langle T \rangle)$ must be at least $1/n$.
    On the other hand, the spreadness condition requires $\mu(\langle T \rangle) \leq 1/k^{|T|} = 1/k^q$.
    Combined, this implies $1/n \leq 1/k^q$, and hence $k \leq n^{1/q}$.

    Now, the high-probability version of the spread lemma (\cite{Rao2020}) guarantees that a $k$-spread set system is $(p, \varepsilon)$-satisfying for $p = O\left(\frac{\log(q/\varepsilon)}{k}\right)$.
    This is known to have an optimal dependence on $\varepsilon$~(e.g., see \cite[Lemma 4]{BCW21}). To obtain $\varepsilon=2^{-\Omega(|K|)}$, one would need to set $p = n^{-\Theta(1/q)} \cdot |K|$ which is vacuous if $|K| = \omega(n^{1/q})$, which is unavoidable in certain cases. This motivates the need for a stronger notion of spreadness, and a stronger spread lemma.
    
    A second difference between the definitions is that the new one is ``closed under conditioning'': if $\mu$ is $(m,k)$-spread, then, for any $\D \subseteq \supp(\mu)$, the conditioned distribution $\mu_{\D}$ is $\left(\frac{m}{\mu(\D)},k\right)$-spread.
    The more refined condition of $(m,k)$-spreadness allows us to express this relation.
    We will see shortly how this property helps us to show that a distribution is a robust daisy.

    \paragraph{The Small-Set Spread Lemma.}
    In \pref{section:spread}, we prove that if a distribution is $(m,k)$-spread, then its support is a satisfying set system.
    \begin{lemma} [``The Small-Set Spread Lemma''; informally stated, see \pref{lemma:spread_distribution}] 
    \label{lemma:spread_distribution-overview} 
        Let $\mu$ be a $(m, k)$-spread distribution, and assume every set in $\supp(\mu)$ has at most $q$ elements.
        Then, for every $\alpha > 2q$, $\supp(\mu)$ is $(p, \varepsilon)$-satisfying with $p = \frac{\alpha}{k}$ and $\varepsilon = \exp\left(- \widetilde{\Omega} \left( \frac{\alpha}{qm} \right)\right)$.
    \end{lemma}

    This lemma shows that $(m, k)$-spreadness is sufficient to obtain the required satisfaction guarantees, and provides a clear target for what parameters we should aim for.
   
    Suppose we can find a dense family $\D$ and a kernel $K$ such that the distribution over the petals, $\mu$, is $(m, k)$-spread with $k=n^{\Theta(1/q)}$ and $m |K| = O(1)$.
    Let $\D' \subseteq \D$.
    As noted above, $\mu_{\D'}$ is $(\frac{m}{\mu(\D')},k)$-spread.
    We could then apply \pref{lemma:spread_distribution-overview} with $\alpha = O(m |K|)$ (which is $O(1)$ by our assumption).
    The lemma then gives us sampling probability $p=O(1/k) = n^{-\Theta(1/q)}$ and failure probability $\varepsilon' = \exp\left(- \widetilde{\Omega} \left( \frac{m|K|}{qm/\mu(\D')} \right)\right) = 2^{-\Omega(\mu(\D') \cdot|K|)}$, implying that the set of petals is $(n^{-\Theta(1/q)}, \varepsilon^{\mu(\D'})$-satisfying for $\varepsilon = 2^{-\Omega(|K|)}$, as needed.
    In the next section we show that we can always achieve such parameters.

    Our proof of the Small-Set Spread Lemma relies on a delicate application of Janson's Inequality.
    We follow the proof strategy of prior spread lemmas \cite{rossman14, ALWZ21}, but crucially leverage our new definition of $(m,k)$-spread. 
    We use this strengthened spread property to obtain a tighter bound on the cumulative dependency among intersecting sets in the family, which in turn yields a smaller failure probability.

    While the above guarantee is stronger than that provided by standard $k$-spread, fortunately, the kernel structure of robust daises enable us to extract such a stronger spread in any distribution, as we discuss next.

\subsection{The spreadness extraction lemma}
    By the discussion above, to obtain the desired RLDC lower bound, the remaining task is as follows.
    We get an arbitrary distribution $\mu$ over a universe $U$ of size $n$, supported on sets of size at most $q$, which for the overview we assume is constant.
    We need to extract an $(m, k)$-spread distribution from $\mu$, and we are allowed to perform the following two operations:
    \begin{enumerate}
        \item \textbf{Puncturing}: We remove a small set of ``problematic'' elements $K \subseteq U$, where $|K| = o(n)$.\footnote{We note that it is sufficient for our purposes that $|K|\leq \delta n$. However, for a constant $q$, we can achieve this better $o(n)$ bound.}
        This $K$ corresponds to the kernel of the robust daisy.
        \item \textbf{Conditioning}: We are allowed to restrict the distribution to a ``well-behaved'' subfamily $\D \subseteq \mathrm{supp}(\mu)$, as long as the density $\mu(\D)$ is large  (e.g., $\mu(\D) \geq 0.99$).
    \end{enumerate}

    In \pref{lemma:large_spread_daisy}, we prove that for any distribution $\mu$, there \emph{always exists} such a set $K$ and a subfamily $\D$ which allow us to achieve the desired spreadness.

    \begin{lemma} (The Spreadness Extraction Lemma; informally stated, see \pref{lemma:large_spread_daisy}) \label{lemma:large_spread_daisy-intro}
        Let $\mu$ be a distribution over $\mathcal{P}(U)$, and assume that every set in $\supp(\mu)$ has at most $q$ elements.
    
        Then, there exists a family of sets $\D \subseteq \supp(\mu)$, and a set $K \subseteq U$ such that the distribution $\mu_{\D}$ punctured by $K$ is $(m, k)$-spread, where:
        \begin{align*}        
            k &=  n^{\Theta(1/q)}, &
            m \cdot |K| &= O(1), &
            \mu(\D) &\geq 0.99 \; .
        \end{align*}
    \end{lemma}
    This extraction process builds upon \cite{gl19} and \cite{Goldreich23}.
    However, it is substantially more involved, as we extract $(m, k)$-spreadness, as opposed to merely a degree bound.

    We next provide a high-level overview of the proof of \pref{lemma:large_spread_daisy-intro}.

    \paragraph{Universe partitioning.}
    We partition the elements in $U$ into $c + 1 =\Theta(q)$ buckets.
    The partitioning is according to the (normalized) \emph{weighted degree} of each element, which we define as:
    \[
        \bar{\mu}(u) = \sum_{S \text{ s.t. } u \in S}\frac{\mu(S)}{|S|} \;.
    \]
    Note that $\bar{\mu}$ is a distribution over $U$.
    That is, $\sum_{x \in U}{\bar{\mu}(x)} = 1$.
    This distribution is equivalent to the following two-step random process:
    first, select a set $S \subseteq U$ according to $\mu$, and then select an element $x \in S$ uniformly at random.
    We remark that if $\mu$ is a uniform distribution over a uniform family $\F$, then $\bar{\mu}(u)$ is the normalized degree of $u$ in $\F$ (the fraction of sets in $\F$ containing $u$) divided by $q$.

    We set a \emph{base step size} $k = n^{1/c}$.
    The buckets are then defined by degree ranges: $B_0$ contains the elements with $\bar{\mu}(u) \leq 1/n$, and for $j \in [c]$, the bucket $B_j$ contains the elements with $\bar{\mu}(u) \in (k^{j-1}/n, k^j/n]$.
    Since $k^c = n$ and the degrees are normalized, this is indeed a partition of the universe.

    \paragraph{The kernel threshold.}
    We choose the kernel $K$ to be all elements with $\bar{\mu}(u) > k^j/n$ for some $j \in [c]$ that we pick later.
    That is, the \emph{threshold} $m =  k^j/n$ is set to be one of the bucket boundaries.
    Since $\sum_{x \in U}{\bar{\mu}(x)} = 1$, this implies $m |K| = O(1)$.
    The question is how to choose this bucket $j$ (which defines $K$ and $m$), such that after puncturing by $K$ and conditioning on an appropriate subfamily $\mathcal{D}$ (as we will see next), the resulting distribution is $(O(m), k)$-spread.

    \paragraph{Token distributions and good boundaries.}
    For any fixed set $S$, the partitioning categorizes its $q$ elements into the buckets.
    We treat these elements as $q$ \emph{tokens} distributed among the $c+1$ buckets, and call this distribution the \emph{token distribution} of $S$.

    Now, we define a key property: we say that a bucket $B_j$ is a \emph{good boundary} for $S$ if the token distribution satisfies the following property: for every $i \in \{0, \dots, j\}$, the set of $i+1$ buckets from $B_{j-i}$ to $B_j$ (inclusive) contains at most $i$ tokens.
    The proof relies on the following two claims:
    \begin{enumerate}
        \item If $B_j$ is a good boundary for every $S \in \D$ (for some $\D \subseteq \supp(\mu)$), then after the removal of $K=\{ u \mid \bar{\mu}(u) > k^j/n\}$ the distribution $\mu_{\D}$ punctured by $K$ is $(O(m), k)$-spread.
        \item There exists $\D \subseteq \supp(\mu)$ with $\mu(\D) > 0.99$ and a bucket $B_j$ which is a good boundary for every $S \in \D$.
    \end{enumerate}
    We next sketch the proofs of these two claims, which form the technical heart of the lemma.

    \paragraph{From good boundaries to spreadness.}
    Recall that to prove spreadness, we need to show that for every $T \subseteq U$, the total mass of sets containing $T$ (denoted by  $\mu(\langle T\rangle)$) is upper bounded by $m/k^{|T|}$.
    The key observation is that: if $T$ \emph{contains an element with small weighted degree}, then  $\mu(\langle T\rangle)$ itself is upper bounded.
    Specifically, let $x \in T$.
    Then, since each set containing $T$ also contains $x$, and each set is of size at most $q$, we have
    $\mu(\langle T\rangle) \leq q \cdot \bar{\mu}(x)$.

    Now, since we removed all elements with $\bar{\mu}(u) > k^j/n$, we can assume $T$ does not contain any such element.
    But, by construction, $B_j$ is a good boundary, hence, by taking $i=|T|-1$, $T$ can contain at most $|T|-1$ elements in the buckets between $B_{j-|T|+1}$ and $B_j$ (inclusive).
    In other words, $T$ must contain an element in $B_{j-|T|}$ (or a lower bucket).
    But all these elements have $\bar{\mu}(u) \leq k^{j-|T|}/n = m/k^{|T|}$.
    Hence, by the above argument, we get the bound $\mu(\langle T\rangle) \leq q m/k^{|T|}$. 

    \paragraph{Abundance of good buckets.}
    We prove that for a any fixed set $S$, at least $c-q$ buckets constitute good boundaries.
    Towards this end, we use a subtle analysis of a \emph{token shifting} process: by iteratively moving tokens from crowded buckets to higher ones, we can argue that the final configuration will have at most $q$ buckets containing any tokens.
    This leaves at least $c-q$ buckets empty, and we prove that these empty buckets must be good boundaries. Taking a sufficiently large $c$ (e.g., $c=100q =\Theta(q)$) yields that almost all the boundaries (say, $99\%$) are good.
    In turn, this implies that there exists (at least) one bucket which is a good boundary for $99\%$ of the sets.
    We choose this bucket to determine $m$ and $K$, and set $\D$ to be the sets for which this bucket is good.

    \paragraph{The number of buckets.}
    It is instructive to discuss how we choose the exact number of buckets, $c$.
    On the one hand, we want this number to be as small as possible, as $c$ determines the exact constant in the exponent of the lower bound.
    Specifically, the lower bound we achieve is $n = \widetilde{\Omega}(k^{1+ \frac{1}{c-1}})$.
    
    On the other hand, it must be sufficiently large; the fraction of good boundaries (out of the total number of buckets) must be larger than the soundness error of the relaxed decoder, which is $1 - \sigma$.
    Otherwise, even if the chosen boundary is good for a large fraction of sets, it is possible that all these sets are ``bad'' --- leading the decoder to a wrong output when choosing them.
    This means we must choose $c$ such that the fraction of good boundaries $\frac{c - q}{c} = 1-q/c$ is larger than $1 - \sigma$; that is, $q/c < \sigma$.
    More precisely, we need $\sigma$ to be bounded away from $q/c$ by some constant independent of $k$.

    Hence, we choose $c = \lceil \frac{q}{\sigma} \rceil+1$, which is the minimal integer satisfying this requirement.
    This parameter choice gives the exact exponent in the lower bound -- namely, $\frac{1}{\lceil q / \sigma \rceil}$.

    \paragraph{The centrality of the spread parameter $k=n^{1/q}$.}
    The choice for the value of $k$, the ``spreadness'' parameter or ``step size'', plays a central role in determining the exact lower bound achieved, even more than the choice of $c$.
    Let us follow how this parameter propagates through the different parts of the proof.

    First, the Small-Set Spread Lemma yields a $(p, \varepsilon)$-satisfying set system with $p=\alpha/k$. Recall we set $\alpha=O(1)$ in this overview, so $p=O(1/k)$.\footnote{More accurately, in the full proof we need to set $\alpha=O(\log^2n \log^2|\Sigma|)$ to compensate for a few minor terms which we do not cover in the overview.}
    This sampling probability $p$ is then used for the robust daisy we extract. By our reduction (\pref{lemma:daisy-to-lowerbound-inf}), this value of $p$ is what gives the final lower bound on the code's length.

    Second, we remark that our current proof technique cannot get an improved value for $k$.
    As we have seen, $k=n^{1/c}$ is the base step size between the buckets, and our proof requires $c > q$.
    If there were fewer than $q + 1$ buckets (i.e., $c \leq q$), it would be possible for a set $S$ to place one of its $q$ tokens in every single bucket.
    In this case, our token-shifting argument would fail to guarantee an empty bucket, and there would be no "good boundaries" at all.

    This $k = n^{\Theta(1/q)}$ barrier is not a coincidence; it appears to be fundamental.
    This exponent is not just a limitation of our proof technique but an inherent feature of the problem, stemming from three different points of view:
    \begin{enumerate}
        \item \textbf{From Construction:} It matches the $n^{1+O(1/q)}$ upper bound, so a better parameter would imply an impossible lower bound.
        \item \textbf{From Combinatorics:} As observed above, $k=n^{\Theta(1/q)}$ is the best spreadness parameter achieved for arbitrary distributions over sets of size $q$.
        \item \textbf{From Our Proof:} Our token-shifting technique, which requires $c > q$ buckets, independently arrives at the same $k=n^{\Theta(1/q)}$ parameter.
    \end{enumerate}
    These perspectives --  the upper bound, the combinatorial limitation of spreadness, and our own extraction method -- solidify the $1/q$ exponent as a central, inherent property of the problem.

\subsection{Organization}
    The rest of the paper is organized as follows.
    In \pref{section:preliminaries} we give standard definitions and notations, including those of RLDCs.
    In \pref{section:spread} we prove the spread lemmas, making the link between spreadness and robust daisies.
    In \pref{section:robust_daisies} we prove the spreadness extraction lemma, and apply it to prove the robust daisy lemma.
    In \pref{section:rldc} we show the reduction from proving a lower bound for RLDCs to the problem of finding a robust daisy, and we apply it to finally prove \pref{thm:main}.
\section{Preliminaries} \label{section:preliminaries}
    We provide notation for set and distributions that we shall use throughout the paper.
    
\subsection{Basic notations}
    \paragraph{Set notation.} Let $U$ be a finite set.
    \begin{itemize} 
        \item For a set $U$, we denote by $\mathcal{P}(U)$ the power set of $U$, i.e., the family of all subsets of $U$. Furthermore, let $\mathcal{P}_{\leq q}(U) = \{S \in \mathcal{P}(U) \;;\; |S| \leq q\}$ denote the family of subsets of size at most $q$. 
    
        \item The \emph{star} of a set $T \subseteq U$, denoted by $\langle T \rangle$, is the family of subsets of $U$ that contain $T$;
        that is, $\langle T \rangle = \{ S \subseteq U \mid S \supseteq T\} \subseteq \mathcal{P}(U)$.

        We use a slight abuse of notation and write $\langle x \rangle$ to denote $\langle \{ x \} \rangle$ for an element $x \in U$.

        \item A family of sets $\F \subseteq \mathcal{P}(U)$ is called \emph{$q$-uniform} if every set $S \in \F$ has size $|S| = q$.

        \item  We use $\mathbb{I}[\cdot]$ to denote the \emph{indicator function}.
        For a set $S$, the indicator function of $S$ is defined as:
        \[
            \mathbb{I}[x \in S] = 
            \begin{cases} 
                1 & \text{if } x \in S \\
                0 & \text{if } x \notin S 
            \end{cases}
        \] 
    \end{itemize}

    \paragraph{Distributions.}
    Let $D$ be a discrete domain.
    
    \begin{itemize}
        \item A \emph{distribution} $\mu$ over $D$ is a function $\mu: D \rightarrow [0, 1]$ such that $\sum_{x \in D}\mu(x)=1$.

        For any subset $A \subseteq D$, the \emph{probability mass} or \emph{density} of $A$ is $\mu(A) = \sum_{x \in A}\mu(x)$.

        \item The \emph{support} of a distribution $\mu$ is the set of elements with non-zero probability, denoted  $\supp(\mu)=\{ x \in D \mid \mu(x) > 0\}$.

        \item The \emph{conditioning} of a distribution $\mu$ to a set $A \subseteq D$ with non-zero density is a distribution over $A$, denoted $\mu_A$, defined for each $x \in A$ by $\mu_A(x)=\frac{\mu(x)}{\mu(A)}$.
        The following straightforward observation relates the probability of events in the conditional distribution to the original distribution.

        \begin{fact} \label{fact:conditioned_subset} 
            For any $B \subseteq D$, $\mu_A(B) =\mu(B)/\mu(A)$.
        \end{fact}
    \end{itemize}

\subsection{Concentration inequalities}
    We use the following version of Janson's inequality.
    \begin{lemma}[{Janson's Inequality \cite[Chapter 8.1] {AS16}}] \label{lem:janson}
        Let $\S \subseteq \mathcal{P}(U)$ be a family of sets.
        Let $W \sim \mathrm{Bin}(U, p)$. For $S \in \S$, let $Z_S$ be the indicator of the event that $S \subseteq W$.
        Let $X = \sum_{S \in \S} Z_S$ and $M = \E[X]$.
        For, $S \neq T \in \S$, let $S \sim T$ iff $S \cap T \neq \emptyset$.
        Define 
        \[
            \Delta = \sum_{(S,T) \;:\; S \sim T} \E[Z_SZ_T] \; .
        \]
        Then, 
        \[
            \Pr[X = 0] \leq \exp(-M + \Delta/2)  \; ,
        \]
        and if $M \leq \Delta$, then
        \[
            \Pr[X = 0] \leq \exp(-M^2/2\Delta) \; .
        \]
    \end{lemma}

\subsection{Relaxed Locally Decodable Codes} \label{section:ecc}
    We recall the formal definition of relaxed locally decodable codes \cite{BGHSV06}.

    \begin{definition} (Relaxed locally decodeable codes) \label{def:rldc}
        Let $C:\{0,1\}^k \rightarrow \Sigma^n$ be an error correcting code.
        A $(q, \delta,\sigma)$-relaxed decoder for $C$ is a randomized procedure $\B$ that on an explicit input $i \in [k]$, and oracle access to $w \in \Sigma^n$, outputs an element of $\{0, 1, \bot\}$, and satisfies the following requirements: 
        \begin{enumerate}
            \item (completeness) If $w = C(x)$ for some $x \in \{0, 1\}^k$, then $\B^{w}(i) = x_i$.\footnote{We remark that a common variation of RLDCs allows the decoder to err with a small probability even on valid codewords.
            Our lower bound holds for such codes as well (assuming they are linear), but for simplicity, we assume perfect completeness throughout.}
            \item (relaxed local decoding) If there exists $x \in \{0, 1\}^k$ such that $\dist(w, C(x)) \leq \delta$, then $\B^{w}(i) \in \{x_i, \bot\}$ with probability at least $\sigma$.
            \item For every input $i$ and oracle access to any $w \in \Sigma^n$, $\B$ makes at most $q$ queries.
        \end{enumerate}
        We say that a $\B$ is non-adaptive if it determines all its queries based on its explicit input (namely, the index to decode) and internal coin tosses, independently of the specific $w$ to which it is given oracle access.
        We refer to $\delta$ as the decoding radius of the decoder, and to $\sigma$ as its soundness probability.
    \end{definition}

    A code $C$ with a $(q, \delta, \sigma)$-relaxed decoder is often referred to as a $(q, \delta, \sigma)$-relaxed locally decodable code.
    
    A relaxed \emph{corrector} for a code is defined analogously to a relaxed decoder, but its objective is to correct any codeword symbol rather than decode a message bit.
    That is, the three requirements in \pref{def:rldc} are extended to any $i \in [n]$, and the corrector's output should be $c_i$ (or $\bot$) instead of $x_i$.

    We say that a code $C:\{0,1\}^k \rightarrow \Sigma^n$ is \emph{linear} if $\Sigma$ is a field and the image of $C$ is a linear subspace of $\Sigma^n$.
\section{Spread lemmas} \label{section:spread}
    In this section, we prove two spread lemmas.
    The first is a spread lemma for \emph{families of sets}.
    We then proceed to extend this lemma, and prove a spread lemma for distributions.

    Let us begin by reiterating the relevant definitions.

    \begin{definition} [Satisfying set system] \label{def:satisfying}
        Let $\mathcal{F}$ be a family of sets over a universe $U$.
        We say that $\mathcal{F}$ is $(p, \varepsilon)$-satisfying if
        \[
            \Pr_{W \sim \textrm{Bin}(U, p)}[\exists S \in \mathcal{F}, S \subseteq W] \geq 1-\varepsilon \:.
        \]
    
        Here, $W \sim \textrm{Bin}(U, p)$ denotes the random set where each element $u \in U$ is chosen to be in $W$ independently with probability $p$.
    \end{definition}
    We remark that the notion of a satisfying set system, and its name, originates from the study of DNF formulas. When a set system is interpreted as a DNF formula, this condition is that the formula has more than a $1 -\varepsilon$ probability of being satisfied on $p$-biased inputs.

    In the following, recall that for $T \subseteq U$, the star of $T$, denoted $\langle T \rangle$, is the family of all subsets of $U$ that contain $T$.
    \begin{definition}[$(m, k)$-spread distributions, reiterating \pref{def:spread-distribution-intro}] \label{def:spread-distribution} 
        Let $\mu$ be a distribution over $\mathcal{P}(U)$, let $k > 1$ and let $m \in (0, 1]$.
        We say that $\mu$ is \emph{$(m, k)$-spread} if for any non-empty set $T \subseteq U$,
        \[
           \mu(\langle T \rangle) \leq \frac{m}{k^{|T|}} \:.
        \]
    \end{definition}

    \begin{definition}[Spread families] \label{def:spread-family} 
        Let $\F$ be a family of sets over a universe $U$, let $k > 1$, and let $m \in (0, 1]$.
        We say that $\F$ is $(m, k)$-spread if the uniform distribution over $\F$ is $(m, k)$-spread.
        That is, if for any non-empty set $T \subseteq U$,
        \[
            \frac{\deg_{\F}(T)}{|\F|} \leq \frac{m}{k^{|T|}} \:,
        \]
        where $\deg_{\F}(T)$ is the number of sets in $\F$ that contain $T$.
    \end{definition}

\subsection{The spread lemma for families of sets}
    We start by proving the spread lemma for fixed set systems.
    
    \begin{lemma} [``The Small-Set Spread Lemma'' - for families of sets] \label{lemma:spread}
        Fix $k > 1$ and $m \in (0, 1]$.
        Let $\mathcal{F}$ be an $(m, k)$-spread family of sets over universe $U$, and assume every set in $\F$ has at most $q$ elements.
        Then, for every $\alpha > 2q$, the family $\F$ is $(p, \varepsilon)$-satisfying with $p = \alpha/k$ and $\varepsilon = \exp \left(-\frac{\alpha}{4 q m} \right)$.
    \end{lemma}

\begin{proof}   
    Our goal is to lower bound $\Pr_{W \sim \text{Bin}(U, p)}[\exists S \in \F, S \subseteq W]$, where $p = \alpha/k$.
    The proof relies on Janson's inequality.

    \paragraph{Setup of Janson's inequality.} 
    For any $S \in \F$, let $\mathcal{Z}_S$ be the indicator value for the event $S \subseteq W$.
    Denote $S \sim T$ if $S, T \in \F$ intersect.
    Define: 
    \begin{equation}
        M = \sum_{S \in \F}{\E[\mathcal{Z}_S]}, \quad \Delta = \sum_{S \sim T}{\E[\mathcal{Z}_S \mathcal{Z}_T]}\:. 
    \end{equation}
    Recall that Janson's inequality states that:
    \begin{enumerate} 
        \item when $\Delta \leq M$: 
            \[
                \Pr_{W \sim \text{Bin}(U, p)}[\exists S \in \F, S \subseteq W] \geq 1 - \exp \left(-M + \Delta/2\right) \geq 1 - \exp(-M/2).
            \]
        \item when $\Delta > M$:
        \[
            \Pr_{W \sim \text{Bin}(U, p)}[\exists S \in \F, S \subseteq W] \geq 1 - \exp \left(-{\frac{M^2}{2\Delta}}\right).
        \]
    \end{enumerate}

    Let $s \leq q$ be the size of the largest set in $\F$.
    It is sufficient to show that
    \[
        \min \left(\frac{M^2}{2\Delta}, \frac{M}{2}\right) \geq \frac{\alpha}{4 s  m} \geq \frac{\alpha}{4 q  m} \:.
    \]

    \paragraph{Uniformity assumption.}
    In what follows, we assume that $\F$ is $s$-uniform.
    This assumption can be made without loss of generality.
    If the family is not uniform, we can increase the universe with a set of ``dummy'' elements and pad each set $S \in \F$ with $s - |S|$ distinct dummies to create a new $s$-uniform family $\F'$.
    This transformation only makes the required condition stricter: for any sample $W$ from the new universe, if a padded set $S' \in \F'$ is fully contained in $W$, then its original counterpart $S$ is necessarily contained in $W$ as well. 
    
    Also note that since each new dummy element is contained only in a single set, this transformation does not affect the spreadness of $\F$.
        
    \paragraph{Estimating $M$.}
    By assumption, every set in $\mathcal{F}$ has $s$ elements, and hence for any $S \in \mathcal{F}$ we have $\E[\mathcal{Z}_S] = p^s$.
    Therefore,
    \begin{equation} \label{eq:mu}
        M = \sum_{S \in \mathcal{F}}{\E[\mathcal{Z}_S]} = |\F| \cdot p^s\:.
    \end{equation}    

    \paragraph{Bounding $\Delta$.}
    For any $t \in [s]$, let $r_t$ be the number of pairs of sets $S, T \in \F$ such that $|S \cap T| = t$.
    Any such $S, T \in \mathcal{F}$ share $t$ vertices, and have $s - t$ unique vertices each.
    Therefore, the probability of sampling all of the vertices of both $S, T$ is
    \[
        \E[\mathcal{Z}_S \mathcal{Z}_T] = p^{t} \cdot (p^{s-t})^2 = p^{2s-t} \:.
    \]
    
    Hence, it follows that
    \[
        \Delta = \sum_{S \sim T}{\E[\mathcal{Z}_S \mathcal{Z}_T]} = \sum_{t=1}^{s}{r_t\cdot p^{2s-t}} = p^{2s} \sum_{t=1}^{s}{r_t \cdot p ^{-t}}\:.
    \]

    We proceed to bound $r_t$.
    For each of the $|\F|$ subsets in the family, there are $\binom{s}{t} \leq s^t$ options to choose the intersection set $R \subseteq S$ of size $t$.
    By the spreadness hypothesis, each such set $R$ is contained in at most $\frac{m}{k^t} |\F|$ sets of $\F$.
    Therefore, in total, $r_t \leq |\F|^2 m \left(\frac{s}{k} \right)^t$, and we get the upper bound
    \[
        \Delta \leq p^{2s} \sum_{t=1}^{s}{|\F|^2 m \left(\frac{s}{k} \right)^t \cdot p ^{-t}} = m p^{2s} |\F|^2 \sum_{t=1}^{s}{ \left(\frac{s}{pk} \right)^t}\:.
    \]

    Next, since $pk = \alpha$ and by assumption $\alpha > 2s \implies(1 - s / \alpha) >1/2$:
    \begin{equation} \label{eq:delta}
        \Delta \leq m p^{2s} |\F|^2 \sum_{t=1}^{s}{\left(\frac{s}{pk} \right)^t} \leq
        m p^{2s} |\F|^2 \sum_{t=1}^{\infty }{ \left( \frac{s}{\alpha} \right)^t} = m p^{2s} |\F|^2 \left( \frac{s/\alpha}{1-s/\alpha} \right) \leq
        m p^{2s} |\F|^2 \frac{2s}{\alpha} \:.
    \end{equation}

    \paragraph{Applying Janson's inequality.}
    Combining \cref{eq:mu} and \cref{eq:delta}, we conclude that
    \[
        \frac{M^2}{2\Delta} \geq \frac{(|\F| p^s)^2}{2 m p^{2s} |\F|^2 \alpha^{-1} 2s} = \frac{\alpha}{ 4s m}\:.
    \]

    Also, note that 
    \[
        \frac{M}{2} = \frac{p^s|\F|}{2} = \frac{\alpha^s|\F|}{2k^s} \geq \frac{\alpha }{2} \cdot \frac{|\F|}{k^s}
    \]
    where the last inequality follows since $\alpha \geq 1$.
    Now, by assumption, there exists a set $S \in \F$ (without dummy elements) with $|S| = s$.
    Hence, from the spreadness of $\F$ when applied to $T=S$, we get $1 \leq \deg(S) \leq |\F| \cdot \frac{m}{k^s} \implies \frac{|\F|}{k^s} \geq \frac{1}{m}$, and then:
    \[
        \frac{M}{2} \geq  \frac{\alpha }{2} \cdot \frac{1}{m} \geq \frac{\alpha}{ 4s m}
    \]
\end{proof}

\subsection{The spread lemma for distributions}
    We use the following helper lemma to transition from distributions to fixed set systems.
    \begin{lemma} \label{lemma:set-large}
        Let $\mu$ be a distribution over a support $R$ with $|R| \geq 2$.
        There exists a non-empty set $A \subseteq R$ such that for every $a \in A$:
        \[
            \mu(a) \geq \frac{1}{2|A| \log(|R|)}
        \]
    \end{lemma}

\begin{proof}
    Let $n = |R|$, and let $p_1 \ge p_2 \ge \dots \ge p_n$ be the sorted probabilities of the elements in $R$.
    Assume for contradiction that for all $k \in \{1, \dots, n\}$, we have $p_k < \frac{1}{2k\log(n)}$.
    
    Summing over all $k$ gives:
    \[
        1 = \sum_{k=1}^n p_k < \sum_{k=1}^n \frac{1}{2 k \log(n)} = \frac{1}{2 \log(n)} \sum_{k=1}^n \frac{1}{k} = \frac{H_n}{2\log(n)}
    \]
    This implies $2\log(n) < H_n$.
    However, for $n \ge 2$, it is a known that $H_n \le \ln (n) + 1 \leq 2\log(n)$, which is a contradiction.
    
    Therefore, there must exist an index $k$ such that $p_k \ge \frac{1}{2k\log(n)}$.
    Let $A$ be the set of the $k$ elements with the largest probabilities.
    Then $|A|=k$, and for any $a \in A$, its probability $\mu(a) \ge p_k$, which satisfies the desired bound.
\end{proof}

Now, we prove the full version of our small-set spread lemma. 
\begin{lemma} [``The Small-Set Spread Lemma''] \label{lemma:spread_distribution} 
    Fix $k > 1$ and $m \in (0, 1]$.
    Let $\mu$ be a $(m, k)$-spread distribution over $\mathcal{P}(U)$ with support $\F$, and assume every set in $\F$ has at most $q$ elements.
    Then, for every $\alpha > 2q$, $\F$ is $(p, \varepsilon)$-satisfying with $p = \alpha /k$ and $\varepsilon = \exp \left( -\frac{\alpha}{8 q m \log |\F|} \right)$.
\end{lemma}

\begin{proof}
    First, observe that if $\F' \subseteq \F$ is $(p, \varepsilon)$-satisfying, then $\F$ is also $(p, \varepsilon)$-satisfying.

    Now, by \pref{lemma:set-large}, there is a subset $\F' \subseteq \F$ such that for every $S \in \F'$:
    \[
        \mu(S) \geq \frac{1}{2 |\F'| \log |\F|} \implies  1 \leq 2 |\F'| \log |\F| \mu(S)
    \]

    We argue that $\F'$ is $(2 \log |\F| m, k)$-spread.
    Let $T \subseteq U$ be a non-empty set.
    Then,
    \begin{align*}
        \frac{\deg_{\F'}{T}}{|\F'|} &= \frac{\sum_{S \in \langle T\rangle} \mathbb{I}[S \in \F']}{|\F'|} \\
        &\leq \frac{\sum_{S \in \langle T \rangle} {2 |\F'| \log |\F| \mu(S)}}{|\F'|} \\
        &= 2 \log |\F| \mu(\langle T \rangle) \leq 2 \log |\F| \cdot m \cdot k^{-|T|}
    \end{align*}
    where in the last inequality we applied the spreadness of $\mu$.

    Now, applying \pref{lemma:spread}, and as every set in $\F'$ has at most $q$ elements, we conclude that $\F'$ is $(p, \varepsilon)$-satisfying (and hence also $\F$) for $p = \alpha / k$ and
    \[
        \varepsilon = \exp \left( - \frac{\alpha}{4 q \cdot 2 \log |\F| m} \right) = \exp \left( -\frac{\alpha}{8 q m \log |\F|} \right)
    \]    
    \end{proof}
\section{Robust daisies} \label{section:robust_daisies}
    In this section, we prove the Robust Daisy Lemma: every distribution over sets of small size contains a dense robust daisy (\pref{def:robust_daisy-intro}).

\begin{lemma} (The Robust Daisy Lemma)  \label{lemma:robust_daisy}
    Let $\mu$ be a distribution over $\mathcal{P}_{\leq q}(U)$.
    For every integer $c > q$ there exists $\D \subseteq \mathcal{P}_{\leq q}(U)$ such that for every $\alpha > 2q$, the conditioned distribution $\mu_\D$ is a $(p, \varepsilon)$-robust daisy with a non-empty kernel $K \subseteq U$ where
    \begin{align*}
        p &= \frac{\alpha}{n^{1/c}}  
        & \varepsilon &= \exp \left(-\frac{\alpha(1-q/c)}{8q^2 \log |\D|} \cdot |K|\right) 
        & |K| & \leq n^{1-1/c}
        & \mu(\D) &\geq  1-\frac{q}{c}.
    \end{align*}
\end{lemma}

    The proof of the lemma is in two steps.
    First, we prove the ``Spreadness Extraction Lemma'', \pref{lemma:large_spread_daisy}.
    We show that it is always possible to make a distribution into a spread one by removing a kernel (``puncturing'') and conditioning on a dense subfamily of the original support.
    This lemma is the main technical part of this section, and is proved in \pref{section:punctured_spreads}.

    The second part of the proof is to apply the small-set spread lemma (\cref{lemma:spread_distribution}) to argue that the extracted spread structure is a robust daisy with the required parameters.

\subsection{The spreadness extraction lemma} \label{section:punctured_spreads}
    We start by formally defining the puncturing operator, and observe its properties.
    
    \begin{definition}[Punctured Distribution] \label{defn:punctured_distribution}
        Let $\mu$ be a distribution over $\mathcal{P}(U)$. The \emph{punctured distribution} of $\mu$ with respect to a set $K \subseteq U$ is a distribution over $\mathcal{P}(U \setminus K)$, denoted  $\mu^{\circ K}$.
        It is defined by the process of first selecting a set $S \subseteq U$ according to $\mu$ and then outputting the set $S \setminus K$.
        The probability of a set $A \subseteq U \setminus K$ is given by
        \[
            \mu^{\circ K} (A) = \sum_{B \subseteq K} {\mu(A \cup B)}.
        \]
        Note that $\sum_{A \in \mathcal{P}(U \setminus K)} \mu^{\circ K} (A) = 1$, and so $\mu^{\circ K}$ is indeed a distribution over $\mathcal{P}(U \setminus K)$.
    \end{definition}

    The following observations about the behavior of puncturing will be useful.

    \begin{observation} \label{obs:punct-start}
        For every distribution $\mu$ over $\mathcal{P}(U)$, a family $\D \subseteq \mathcal{P}(U)$ and a set $K \subseteq U$, the following holds:
        \begin{enumerate}
            \item If $\supp(\mu) = \D$ then $\supp(\mu^{\circ K}) =\D \setminus K = \{ S \setminus K \mid S \in \D\}$.
            \item If $\mu$ is supported on sets of size at most $q$, then so is $\mu^{\circ K}$.
            \item $\mu^{\circ K}(\D \setminus K) \geq \mu(\D)$.
            \item For every $T \subseteq U \setminus K$, it holds that $\mu^{\circ K}(\langle T \rangle) = \mu(\langle T \rangle)$.
        \end{enumerate}
    \end{observation}

    \begin{proof}
        Items 1-3 follow directly from the definitions.
        Item 4 follows by carefully expanding the definitions:
        \[
            \mu^{\circ K}(\langle T \rangle) = \sum_{A \text { s.\ t } T \subseteq A \subseteq U \setminus K}{\mu^{\circ K}(A)} = \sum_{A \text { s.\ t } T \subseteq A \subseteq U \setminus K}{\sum_{B \subseteq K}{ \mu(A \cup B)}}
        \]
        Now observe that the double summation is the same as summing over all subsets of $U$ that contain $T$.
    \end{proof}

    We are now ready to state and prove the spreadness extraction lemma.
    
\begin{lemma} (The Spreadness Extraction Lemma) \label{lemma:large_spread_daisy}
    Let $\mu$ be a distribution over $\mathcal{P}_{\leq q}(U)$, where $|U| = n$.
    Let $c > q$ be any integer.
    
    Then, there exists $j \in [c]$, a family of sets $\D \subseteq \mathrm{supp}(\mu)$, and a non-empty set of elements $K \subseteq U$ such that the distribution $(\mu_\D)^{\circ K}$ is $(m, n^{1/c})$-spread, where:
    \begin{align*}        
        m &= \frac{q}{\mu(\D)} \cdot \frac{n^{j/c}}{n}, &
        |K| & \leq \frac{n}{n^{j/c}}, &
        \mu(\D) &\geq 1 - \frac{q}{c} \;.
    \end{align*}
\end{lemma}
    We note that when $q$ and $c$ are treated as constants, the parameters satisfy the relation $m \cdot|K| = O(1)$, which is needed for the robust daisy extraction.
    The specific value of the index $j$ is not crucial beyond its role in defining $m$ and $K$.

\begin{proof}
    The proof proceeds by partitioning the universe $U$ into $c+1$ buckets based on the \emph{weighted degree} of each element, $\bar{\mu}$.
    We define $\bar{\mu}$ for each $x \in U$ by:\footnote{Recall that $S \in \langle x \rangle \iff x \in S$.}
    \[
        \bar{\mu}(x) = \sum_{S \in \langle x \rangle}\frac{\mu(S)}{|S|} \;.
    \]
    Note that $\bar{\mu}$ is a distribution over $U$.
    That is, $\sum_{x \in U}{\bar{\mu}(x)} = 1$.
    This distribution is equivalent to the following two-step random process:
    first, select a set $S \subseteq U$ according to $\mu$, and then select an element $x \in S$ uniformly at random.
    
    The construction identifies a threshold defined by one of the buckets, indexed by $j$.
    The punctured set $K$ consists of all elements with a weighted degree above this threshold.

    \paragraph{Construction.}
    Let $k = n^{1/c}$.
    We partition the universe~$U$ into $c\;+\;1$ buckets, 
    $\{ B_0, B_1,\ldots, B_c\}$, based on the weighted degree $\bar{\mu}(u)$.
    Define $B_0 = \left\{u \in U \mid \bar{\mu}(u) \leq \frac{1}{n}  \right\}$ and for each $j \in [c]$:
    \begin{equation}
        B_j = \left\{u \in U \mid \bar{\mu}(u) \in \left(\frac{k^{j-1}}{n}, \frac{k^{j}}{n} \right] \right\}\;.
    \end{equation}
    Since $k^c = n$, the upper bound for $\bar{\mu}(u)$ in $B_c$ is $1$.
    As $\bar{\mu}(u) \leq 1$ for any $u \in U$, the buckets
    
    $\{ B_0, B_1, \ldots, B_c\}$ form a partition of $U$.
    
    For any indices $i \leq \ell$, we denote $B_{[i, \ell]} = \cup_{j = i}^{\ell} B_j$.

    The core of our argument relies on identifying a ``well-behaved'' boundary between buckets.
    We formalize this in the following definition.
    \begin{definition}\label{def:good_bucket} (Good Boundary)
        Let $j \in [c]$ and $S \subseteq U$.
        We say that the bucket $B_j$ is a \emph{good boundary} for the set $S$ if for every $i \in \{0, \dots, j \}$, we have $|S \cap B_{[j - i, j]}| \leq  i$.\footnote{Note that for $i \geq q$, this condition becomes trivial, as $|S| \leq q$ for every $S \in \supp(\mu)$.}
    \end{definition}
    Note that by definition, $B_0$ cannot be a good boundary.

    The lemma now follows from the three claims below, and setting $K = B_{[j+1, c]}$. It is possible that $B_{[j+1, c]}$ is empty, in which case we set $K = \{u\}$ for an arbitrary $u \in U$. This does not affect any of our arguments. 
    
    \begin{claim} \label{claim:good-boundary-existance}
        There exists an index $j \in [c]$ and a family $\D \subseteq \mathcal{P}(U)$ such that $B_j$ is a good boundary for every $S \in \D$, and $\mu(\D) \geq 1 - \frac{q}{c}$.
    \end{claim}

    \begin{claim} \label{claim:boundary-to-spread}
        If $B_j$ is a good boundary for every $S \in \D$, then the distribution $(\mu_{\D})^{\circ B_{[j+1, c]}}$ is $\left(\frac{q}{\mu(\D)} \cdot \frac{k^j}{n} \;,\; k\right)$-spread.
    \end{claim}

    \begin{claim} \label{claim:punc-size}
        For every index $j \in [c]$, the set of elements with high weighted degree is small: $|B_{[j+1, c]}| \leq \frac{n}{k^j}$.
    \end{claim}
    
    We begin by proving the claim that motivates the definition of good boundaries: if a bucket is a good boundary for some sets $\D$, then by conditioning on $\D$, and then removing the elements ``above'' this bucket, we get a spread distribution. 
    \begin{proof} [Proof of \pref{claim:boundary-to-spread}]
    Assume that $B_j$ is a good boundary for every $S \in \D$.
    We prove that the distribution $(\mu_{\D})^{\circ B_{[j+1, c]}}$ is $\left(\frac{q}{\mu(\D)} \cdot \frac{k^j}{n} \;,\; k\right)$-spread.

    \paragraph{Restating the goal.}
    The distribution $(\mu_{\D})^{\circ B_{[j+1, c]}}$ is over subsets of 
    \[
        U \setminus B_{[j+1, c]} = B_{[0, j]} =  \left\{ u \in U \mid \bar{\mu}(u) \leq \frac{k^{j}}{n} \right\} \:.
    \]
    Hence, to show it is spread, we need to show that for every non-empty set $T \subseteq B_{[0, j]}$ we have 
    \[
        (\mu_{\D})^{\circ B_{[j+1, c]}}(\langle T \rangle) \leq \frac{q}{\mu(\D)} \cdot \frac{k^j}{n} \cdot k^{-|T|}
    \]
    
    By \cref{obs:punct-start} (item $4$) and the definition of $\mu_{\D}$, this is equivalent to:
    \[
        \mu(\langle T \rangle) \leq q \cdot \frac{k^{j - |T|}}{n} \:.
    \]
    To show this, we first bound $\bar{\mu}(x)$ for some $x \in T$.
    Then, we use this bound to bound $\mu(\langle T \rangle)$.

    \paragraph{Bounding $\bar{\mu}(x)$ for $x \in T$.}
    We may assume there exists a set $S \in \D$ with $T \subseteq S$; otherwise, $\mu_{\D}(\langle T \rangle) = 0$ and the spreadness requirement holds trivially.
    
    First, we show that $|T| \leq j$. 
    Since $B_j$ is a good boundary, $|S \cap B_{[j - i, j]}| \leq i$ for every $i \leq j$.
    Specifically, setting $i = j$ we have $|S \cap B_{[0, j]}| \leq j$.
    Now, because $S$ is a subset of $B_{[0, j]}$, this simplifies to $|S| \leq j$, and therefore $|T| \leq j$.

    Next, we apply the same inequality with $i = |T|-1$.
    Note that $|T| \leq j$ ensures $i \leq j$.
    This yields:
    \[
        \left|T \cap B_{[j - |T| + 1, j]}\right| \leq |T| - 1 \: .
    \] 
    The inequality implies that at most $|T|-1$ elements of $T$ belong to the set $B_{[j - |T| + 1, j]}$.
    By the pigeonhole principle, at least one element $x \in T$ must lie outside this set.
    Since all elements of $T$ are in $B_{[0, j]}$, $x$ must be in $B_{[0, j]} \setminus B_{[j - |T| + 1, j]} = B_{[0, j - |T|]}$.
    By definition, this means $\bar{\mu}(x) \leq \frac{k^{j - |T|}}{n}$.

     \paragraph{Bounding $\mu(\langle T \rangle)$.}
     First, observe that for any $x \in U$, we have $\mu(\langle \{ x\} \rangle) \leq q \cdot \bar{\mu}(x)$.
     This is because every set in the support of $\mu$ has at most $q$ elements, and then by definition $\bar{\mu}(x) = \sum_{S \in \langle x \rangle}\frac{\mu(S)}{|S|} \geq \frac{ \mu(\langle x \rangle)}{q}$.

     In addition, for every $T' \subseteq T$  we have $\mu(\langle T \rangle) \leq \mu(\langle T' \rangle)$, since in the left side of the inequality we sum up over (possibly) fewer sets.
     Hence, for $T' = \{x \}$ we get $\mu(\langle T \rangle) \leq \mu(\langle \{ x \} \rangle)$.
     Together with the bound on $\bar{\mu}(x)$, this implies that:
    \[
        \mu(\langle T \rangle) \leq \mu(\langle x \rangle) \leq q \bar{\mu}(x) \leq q \cdot \frac{k^{j - |T|}}{n} \:.
    \]

    \end{proof}

    The proof of \pref{claim:good-boundary-existance} relies on the following combinatorial claim and a standard averaging argument.
    \begin{claim}\label{claim:good_sets}
        For any set $S \subseteq U$ with $|S| \leq q$, there are at least $c - q$ indices $j \in [c]$ for which $B_j$ is a good boundary for $S$.
    \end{claim}

    \begin{proof}
    Let $S \subseteq U$ with $|S| \leq q$.
    To prove the claim, we show that there are at most $q$ ``bad boundaries'' for $S$, which implies that at least $c-q$ boundaries must be good.
    The proof uses a ``token-shifting'' argument.
    We conceptualize the elements of $S$ as tokens distributed among the buckets.
    Initially, each bucket is assigned a number of tokens that is equal to the number of elements of $S$ that are in this bucket.
    The total number of tokens is therefore $|S| \leq q$.
    We then apply an iterative redistribution process that shifts these tokens between buckets to produce a final configuration.
    
    The core of the argument is to show that in this final configuration, if a bucket contains no tokens, then it must be a good boundary.
    We now proceed with the formal description of the redistribution process and its analysis.

    \paragraph{The Token-Shifting Process.}
    The process is defined as follows.    
    We start with the initial token counts $w_j =|B_j \cap S|$ for $j \in [c]$.
    As long as there is any bucket $B_j$ that contains two or more tokens ($w_j \geq 2$), we apply the following rule: move one token from bucket $B_j$ to the bucket above it, $B_{j+1}$.
    That is, update the counts: $w_j \gets w_j - 1$ and $w_{j+1} \gets w_{j+1} + 1$.
    If $j = c$ and there is not bucket above it, only remove a token from $B_c$.
    The process ends when every bucket has one or zero token.

    \paragraph{Loop Invariants.}
    We claim that the following three invariants hold throughout the operation of the reassignment process.
    \begin{enumerate}
        \item The total number of tokens is at most $|S| \leq q.$
        That is, $\sum_{j=1}^c{w_j} \leq q$.
        This is because the total number of tokens starts at $|S|$ and can only decrease if bucket $B_c$ is chosen.
        \item If at some point a bucket has non-zero tokens, then it will never have zero tokens later.
        This is because we remove tokens from a bucket only if it has at least $2$ tokens.
        \item If bucket $B_j$ has no tokens, then for any $i$, the total number of tokens in the $i$ buckets below $B_j$ is at least the number of elements of $S$ in those buckets.
        That is, $\sum_{j'=j-i}^{j}{w_{j'}} \geq \sum_{j'=j-i}^{j}{|B_{j'} \cap S|}$.
        This is because a bucket with a non-zero weight will never become zero, and the total number of tokens in a sequence of buckets can decrease only when the bucket just above them gains a token.
    \end{enumerate}

    \paragraph{Post-loop claims.}
    From these invariants, we  deduce the following claims about the weights at the end of the loop:
    \begin{enumerate}
        \item There are at most $q$ non-zero buckets.
        Each bucket ends with a weight of at most $1$, and by the first invariant the total number of tokens is at most $q$.
        \item If a bucket $B_j$ has no tokens, then it is a good boundary.
        Otherwise, there is some $\ell$ such that $\left|S \cap B_{[j-i,j]}\right| > i$.
        But according to the third invariant, since $B_j$ contains no tokens, the total number of tokens in the $i$ buckets below $B_j$ is at least $\left|S \cap B_{[j-i,j]}\right| > i$.
        This is a contradiction, since after the loop ends each of these buckets has at most one token, and $B_j$ contains no tokens.
    \end{enumerate}
    Together, these claims imply that there are at least $c - q$ good boundaries for $S$.
    \end{proof}

    We can now complete the proof of \pref{claim:good-boundary-existance}
    \begin{proof} [Proof of \pref{claim:good-boundary-existance}]
        Let $G(S) \subseteq [c] $ be the set of good boundaries for a set $S$.
        By \pref{claim:good_sets}, $|G(S)| \geq c - q$ for every set $S$, and hence the expected number of good boundaries for a set $S \sim \mu$ is also lower bounded:
        \[
            \mathbb{E}_{S \sim \mu}[|G(S)|] \geq c-q
        \]
        On the other hand, by linearity of expectation, the expectation can be expressed as:
        \begin{align*}
            \mathbb{E}_{S \sim \mu}[|G(S)|] &= 
            \mathbb{E}_{S \sim \mu} \left[ \sum_{j=1}^c \mathbb{I}[j \in G(S)] \right] \\
            &= \sum_{j=1}^c \mathbb{E}_{S \sim \mu} [\mathbb{I}[j \in G(S)]] \\
            &= \sum_{j=1}^c \mu\left(\{S \mid j \in G(S)\}\right)
        \end{align*}
    
        Combining these facts implies there must exist an index $j^* \in [c]$ such that
        \[
        \mu(\{S \mid j^* \in G(S)\}) \geq \frac{c-q}{c} \:.
        \]
        We define the family $\D$ as:
        \[
            \D = \{ S \in \mathrm{supp}(\mu) \mid B_{j^*} \text{ is a good boundary for } S \}
        \]        
        For this choice, we have:
        \[
            \mu(\D) \geq \frac{c-q}{c} = 1 - \frac{q}{c}    
        \]
    \end{proof}
    
    Finally, we prove the simplest claim, \pref{claim:punc-size}.
    \begin{proof} [Proof of \pref{claim:punc-size}]
        By the definition of the buckets, every element $u \in B_{[j+1, c]}$ has a weighted degree $\bar{\mu}(u) > \frac{k^j}{n}$.
        Summing this lower bound over all elements in $B_{[j+1, c]}$ yields:
        \[
            \sum_{u \in B_{[j+1, c]}}{\bar{\mu}(u)} > \left|B_{[j+1, c]}\right| \cdot \frac{k^j}{n} \: .
        \]
        On the other hand, since $\bar{\mu}$ is a probability distribution over the universe $U$, the sum of probabilities over any subset of $U$ is at most $1$:
        \[
            \sum_{u \in B_{[j+1, c]}}{\bar{\mu}(u)} \leq \sum_{u \in U}{\bar{\mu}(u)} = 1 \:.
        \]
        Combining these two inequalities gives:
        \[
            \left|B_{[j+1, c]}\right| \cdot \frac{k^j}{n} < 1 \:,
        \]
        and rearranging the terms finishes the proof.
    \end{proof}
\end{proof}

\subsection{Large robust daisies } \label{section:spreadness_to_robustness}
    We next use the spread lemma for distributions (\pref{lemma:spread_distribution}) to argue that if $\mu^{\circ K}$ is spread, then $\mu$ is a robust daisy with kernel $K$.

    First, we make the following observation:
    \begin{observation} \label{observation:cond-spread}
        Fix $k > 1$ and $m \in (0, 1]$.
        Let $\mu$ be a $(m, k)$-spread distribution over $\mathcal{P}(U)$ and let $\D \subseteq \supp(\mu)$.
        Then, $\mu_{\D}$ is $\left(\frac{m}{\mu(\D)}, k\right)$-spread.
\end{observation}

\begin{proof}
    Let $T \subseteq U$ be a non-empty set.
    By the definition of conditioned distribution, and by the spreadness of $\mu$, we have:
    \[
        \mu_{\D}(\langle T \rangle) \leq \frac{\mu(\langle T \rangle)}{\mu(\D)} \leq \frac{1}{\mu(\D)} \cdot \frac{m}{k^{|T|}}
    \]
\end{proof}
Armed with this observation, we show that a distribution which is spread outside a kernel is also a robust daisy. 
\begin{lemma} \label{lemma:spread-to-robust}
    Fix $k > 1$ and $m \in (0, 1]$.
    Let $\mu$ be a distribution over $\mathcal{P}_{\leq q}(U)$.
   Suppose $\mu^{\circ K}$ is $(m, k)$-spread for some $K \subseteq U$.
   Then, for every $\alpha > 2q$, $\mu$ is a $(p, \varepsilon)$-robust daisy with kernel $K$, $p = \alpha/k$ and $\varepsilon = \exp \left(-\frac{\alpha}{8 q m \log|\supp(\mu)|} \right)$.
\end{lemma}  

\begin{proof}
    To prove that $\mu$ is a $(p, \varepsilon)$-robust daisy, we need to show that for every $\D \subseteq \supp(\mu)$, the family $\D \setminus K :=\{ S \setminus K \mid S \in \D \}$ is $(p, \varepsilon^{\mu(\D)})$-satisfying.

    By assumption, $\mu^{\circ K}$ is $(m, k)$-spread.
    Therefore, by \pref{observation:cond-spread}, the distribution $(\mu^{\circ K})_{\D \setminus K}$ is $\left(\frac{m}{\mu^{\circ K}(\D \setminus K)}, k\right)$-spread.\footnote{We remark that, perhaps unintuitively, it is not always true that $(\mu^{\circ K})_{\D \setminus K} = (\mu_{\D})^{\circ K}$.}

    Now, by \pref{lemma:spread_distribution}, for every $\alpha > 2q$ we have that $\supp((\mu^{\circ K})_{\D \setminus K}) = \D \setminus K$ is $(p, \varepsilon')$-satisfying with $p = \alpha/k$ and
    \[
        \varepsilon' =\exp \left(-\frac{\alpha}{8 q (m /\mu^{\circ K}(\D \setminus K)) \log |\D  \setminus K|} \right)
        \leq \exp \left(-\frac{\alpha\cdot \mu(\D)}{8q m \log |\supp(\mu)|} \right) 
        = \varepsilon^{\mu(\D)} \; .
    \] 
    where the inequality follows from the observation that $\mu^{\circ K}(\D \setminus K) \geq \mu(\D )$ (\cref{obs:punct-start}) and since $|\supp(\mu)| \geq |\D| \geq |\D  \setminus K|$. 
\end{proof}

We can now prove the main result of this section.

\begin{proof} [Proof of \pref{lemma:robust_daisy}]
    The lemma follows from the Punctured Spread Distribution Lemma (\pref{lemma:large_spread_daisy}), combined with the fact that spreadness implies robustness (\pref{lemma:spread-to-robust}).
    
    Let $c > q$.
    By \pref{lemma:large_spread_daisy}, there exists $j \in [c]$, $\D \subseteq \mathcal{P}_{\leq q}(U)$ and a non-empty $K \subseteq U$ such that $(\mu_{\D})^{\circ K}$ is a $(m, n^{1/c})$-spread distribution, where:
    \begin{align*}
        m &= \frac{q}{\mu(\D)} \cdot \frac{n^{j/c}}{n}, &
        |K| & \leq \frac{n}{n^{j/c}}, &
        \mu(\D) &\geq 1 - \frac{q}{c}\;.
    \end{align*}

    Suppose $\alpha > 2q$.
    Now, we apply \pref{lemma:spread-to-robust} to conclude that $\mu_\D$ is a $(p, \varepsilon)$-robust daisy with $p = \alpha / k = \alpha n^{-1/c}$ and
    \begin{align*}
        \varepsilon &= \exp \left(-\frac{\alpha}{8qm\log|\D|} \right)  & (\supp(\mu_{\D}) = \D) \\
        &= \exp \left(-\frac{\alpha \cdot \mu(\D) \cdot n}{8q \cdot q \cdot n^{j/c} \cdot \log |\D|}\right) & \left(\text{plugging in $m$}\right)  \\
        & \leq \exp \left(-\frac{\alpha \mu(\D) \cdot |K|}{8q^2 \log |\D|}\right)  & \text{(since $|K| \leq \frac{n}{n^{j/c}}$)}  \\
        & \leq \exp \left(-\frac{\alpha(1-q/c)|K|}{8q^2 \log |\D|}\right). & \left(\text{since $\mu(\D) \geq 1-q/c$}\right)
    \end{align*}

    Finally, since $j \geq 1$, we have $|K| \leq \frac{n}{n^{j/c}} \leq n^{1-1/c}$, as required.
\end{proof} 
\section{RLDC lower bounds} \label{section:rldc}
    This section proves the main result (\cref{thm:main}) of the paper.
    
    In \pref{section:lb-overview}, we provide a detailed overview that motivates and discusses the definition of robust daisies and the ``global sampler'' strategy.
    In \pref{section:lb-struct}, we formally prove that if a relaxed decoder possesses the structure of robust daisies, this implies a lower bound on its block length.
    \pref{section:lb-overview} and \pref{section:lb-struct} are independent of the previous sections, and we encourage readers unfamiliar with the work of \cite{gl19} to read them first.
    
    Finally, in \pref{section:lb-general}, we prove the main result, by using the Robust Daisy Lemma to show that any relaxed decoder can be transformed into one that has the structure of robust daisies.
    
\subsection{Overview} \label{section:lb-overview}
    At its core, the proof is a compression-based, information-theoretic argument: one cannot recover $k$ bits of information by observing fewer than $k$ bits, except with small probability.
    
    For the rest of the overview, let  $C: \{0,1\}^k \rightarrow \{0,1\}^n$ be a  $q$-query RLDC with a  non-adaptive decoder, constant decoding radius $\delta$, and soundness probability $\sigma$.
    We show that if the block length $n$ is too small as a function of the message length $k$, it is possible to recover, with a high probability, the entire $k$-bit message by querying fewer than $k$ bits of the corresponding codeword.

    \paragraph{Global sampler.}
    A \emph{global sampler} is a probabilistic algorithm with oracle access to a \emph{valid codeword} $C(x)$.
    It samples each bit of the codeword independently with some probability $p$, and its goal is to recover the entire message $x$.
    With a high probability, the sampler queries $\Theta(pn)$ bits in total.
    Therefore, if it succeeds in recovering the message $x$, it implies a bound  of $pn=\Omega(k)$.
    For concreteness, a sampling probability of exactly $p=n^{-1/q}$ would give the bound $n^{1-\frac{1}{q}} = \Omega(k)$, which implies $n = \Omega\left(k^{1+\frac{1}{q-1}}\right)$.

    The global sampler leverages the relaxed decoder of $C$ to recover the bits of $x$.
    To recover the $i$-th bit, $x_i$, the global sampler aims to simulate the relaxed decoder for index $i$.
    The challenge is that the sampler's random samples must be sufficient to decode all $k$ indices \emph{simultaneously}, whereas the relaxed decoder for each index $i$ may query an arbitrary set, according to a distribution depending on $i$.

    \paragraph{Warm-up: satisfying set systems.}
    Consider a non-adaptive relaxed decoder for a code $C:\{0, 1\}^k \to \{0, 1\}^n$.
    To decode $x_i$, the decoder performs as follows:
    it picks a set $S$ according to some distribution $\mu_i$, queries the indices in $S$, and computes its output \emph{deterministically} based on its sampled \emph{local view}.\footnote{We can assume that this computation is deterministic, since the decoder is non-adaptive and must never err on valid codewords.}

    If the global sampler happens to draw a set of samples that contains a full set from the support of $\mu_i$, it can recover $x_i$ by invoking the decoder's logic on that local view.
    Since the global sampler's input is a valid codeword, the relaxed decoder is guaranteed to decode correctly. 

    This idea is captured by the well-known notion of a \emph{Satisfying Set System}.
    In the following, $W \sim \textrm{Bin}(U, p)$ denotes a random subset of $U$ where each element is included independently with probability $p$.

    \begin{definition} [Satisfying set system, restating \pref{def:satisfying}]
        Let $\mathcal{F}$ be a family of sets over a universe $U$.
        We say that $\mathcal{F}$ is $(p, \varepsilon)$-satisfying if
        \[
            \Pr_{W \sim \textrm{Bin}(U, p)}[\exists S \in \mathcal{F}, S \subseteq W] \geq 1-\varepsilon \:.
        \]
    \end{definition}

    This suggests that our goal is to show that each family $\F_i := \supp(\mu_i)$ is $(p, \varepsilon)$-satisfying for a small $\varepsilon$ (e.g., $\varepsilon \ll 1/k$, to allow for a union bound over all $k$ indices).
    However, it is unreasonable to expect the query families $\F_i$ to satisfy such a strong property.
    The decoder's queries may be highly non-uniform; for instance, to decode index $i$, the decoder might always query index $i$ itself — a common feature in known constructions.
    In this case, the global sampler must happen to sample index $i$ to capture any full set $S \in \F_i$.
    Since this occurs with probability $p$, the sampler would effectively need to query almost all of $C(x)$ to recover all of $x$, yielding a trivial bound.
    We must therefore handle these ``heavy indices'' differently.

    \paragraph{The actual global sampler: robust daisies.}
    Our strategy is to \emph{guess} the values of $C(x)$ at the few ``heavy'' indices instead of trying to sample them.
    Let us denote the set of these heavy indices by a kernel $K \subseteq [n]$.
    To obtain a full local view for a query set $S$, the global sampler no longer needs to sample all of $S$; it only needs to sample the ``light'' elements in $S \setminus K$.

    We call the sets $\{ S \setminus K \mid S \in \F_i\}$ the \emph{petals}.
    The intuition is that it might be sufficient for the family of petals to be satisfying, rather than the family of full query sets.

    If the guess for the heavy indices is correct, each local view we construct is consistent with the valid codeword $C(x)$.
    Therefore, since the decoder never errs on valid codewords, applying its logic to any of these views will yield the correct output bit.

    But what if our guess for the bits in $K$ is incorrect?
    In this case, the local views we construct are not consistent with $C(x)$, but rather with a corrupted version of it.
    The number of corruptions is at most $|K|$.
    If $|K|$ is smaller than the decoding radius of the code, then the decoder's soundness guarantee holds.
    This means that for any guess, at least a large fraction (i.e., $\sigma$, where ``fraction'' is measured with respect to the distribution $\mu_i$) of the query sets will lead the decoder to output the correct symbol or the special rejection symbol $\bot$.

    This leads to the actual algorithm for the global sampler.
    To recover each bit $x_i$:
    \begin{enumerate}
    \item Sample the codeword $C(x)$ by picking each index i.i.d.\ with probability $p$.
    \item Iterate over all $2^{|K|}$ possible guesses for the values of the bits at the heavy indices in $K$.
        \item For each guess, identify all petals $S \setminus K$ that were fully contained in the initial sample.
        \item For each such petal, form a complete local view using the sampled values and the current guess.
        Feed all these views to the decoder's logic.
        \item If all these simulated local views result in the same output bit, output that bit for $x_i$.
        Otherwise, continue.
    \end{enumerate}

    The algorithm is guaranteed to produce an output - for the correct guess on $K$.
    By the decoder's completeness, every local view consistent with $C(x)$ leads to the correct output bit, ensuring a consensus.
    An incorrect output for $x_i$ can only occur if, for some incorrect guess, \emph{all} sampled petals happen to correspond to query sets that mislead the decoder.

    Conversely, for any guess, the algorithm avoids an error as long as at least one sampled petal corresponds to a ``good'' query set — one that would lead the decoder to output the correct bit or $\bot$.
    By the soundness guarantee, a large fraction of query sets are good.
    Therefore, we expect that our sampler will likely hit one of the many corresponding ``good'' petals.

    This motivates our central definition of a \emph{Robust Daisy}.
    These structures differ from naive satisfying set systems in two ways. First, they allow for a kernel $K$ to be handled separately.
    Second, they require that \emph{any} sufficiently large sub-family is also satisfying. 
    This ensures that even if we restrict our attention to the ``good'' sets (which form a large subfamily), we are still guaranteed to sample one of their petals. 

    For simplicity, our definition requires this property for \emph{all} sub-families, a stronger condition we can achieve without extra cost.

    \begin{definition} (Robust daisy, restating \pref{def:robust_daisy-intro})
        A distribution $\mu$ over $\mathcal{P}(U)$ is a \emph{$(p, \varepsilon)$-robust daisy} with kernel $K \subseteq U$, if, for every $\D \subseteq \supp(\mu)$, the family of petals $\D \setminus K :=\{ S \setminus K \mid S \in \D \}$ is $(p, \varepsilon^{\mu(\D)})$-satisfying.
        That is, if:
        \[
            \Pr_{W \sim \text{Bin}(U, p)}[\exists \; S \in \D, S \subseteq K \cup W] \geq 1-\varepsilon^{\mu(\D)} \: .
        \]
    \end{definition}

    \paragraph{From robust daisies to a global sampler.}
    Now, suppose that for some $i \in [k]$, the query-set distribution of the relaxed decoder for index $i$, denoted $\mu_i$, is a $(p, \varepsilon_i)$-robust daisy with kernel $K_i$.
    Recall that at least a $\sigma$ fraction of the petals are ``good'' (i.e., they lead to a correct output or $\bot$).
    Hence, by the discussion above, for any specific guess of the kernel values, the probability that the global sampler \emph{fails} to sample a ``good'' petal is at most $\varepsilon_i^\sigma$.
    By taking a union bound over all $2^{|K_i|}$ possible guesses, the probability that the global sampler fails to find a good petal for \emph{any} of the guesses, and hence outputs a wrong bit for $x_i$, is at most $2^{|K_i|} \cdot \varepsilon_i^\sigma$.

    Taking another union bound over all $i \in [k]$, and assuming that every $\mu_i$ is a robust daisy, gives that with probability at least $1 - \sum_{i \in [k]}{\varepsilon_i ^{\sigma} \cdot 2^{|K_i|}}$ the global sampler recovers \emph{all} of the $k$ bits of $x$, without any mistakes.
    This brings us to the exact requirement we need from the robust daisy, and the formal proof of the reduction, which we present in \pref{section:lb-struct}.

    However, what if the distributions $\mu_i$ are not robust daisies?
    A simple argument shows that it is enough to extract a robust daisy of large density from each $\mu_i$.
    Namely, if we find $\D \subseteq \supp(\mu)$ such that $\mu_{\D}$ is a robust daisy, we can modify the relaxed decoder such that instead of sampling a set $S \sim \mu$ and querying the indices of $S$, the modified decoder samples a set $S \sim \mu_{\D}$.
    Assuming that $\mu(\D)$ is large enough compared to $\sigma$ (specifically, if $\mu(\D) > 1 - \sigma$), then the modified decoder has non-trivial soundness, and can be used by the global sampler. 
    We give a formal proof of this argument in \pref{section:lb-general}, and use it to complete the proof of \pref{thm:main}. 

\subsection{Lower bound for RLDCs with structured decoders} \label{section:lb-struct}
    In this section, we prove a lower bound for RLDCs with relaxed decoders whose query set distributions are robust daisies. 

    \begin{lemma} \label{lemma:daisy-to-lowerbound}
        Let $C :\{0,1\}^k\rightarrow \Sigma^n$ be a non-adaptive $(q, \delta, \sigma)$-RLDC, and let $p$ satisfy $\frac{3\ln(n)}{n} < p < 1$. 
        For each $i \in [k]$, let $\mu_i$ be the decoder's query distribution for index $i$.
    
        Assume that for every $i \in [k]$, $\mu_i$ is a $(p, \varepsilon_i)$-robust daisy with a kernel $K_i \subseteq [n]$ such that
        \begin{equation} \label{eqn:robust_condition}
            |K_i| \leq \delta n \quad \text{and} \quad \varepsilon_i^{\sigma} \leq \frac{1}{3k \cdot |\Sigma|^{|K_i|}} \: .
        \end{equation}
        Then $k \leq 2pn \cdot \log|\Sigma|$.

    \end{lemma}

    \begin{proof}
        We give a formal description of the global sampler, $\G$, in \cref{algorithm:global_decoder}.

        \paragraph{Notation.} In the description of $\G$ we use the following notation.
        Let $i \in [k]$, and let $\B_i$ be the corresponding relaxed decoder with query set distribution $\mu_i$.
        Let $\D_i = \supp(\mu_i)$.
        On a codeword $y$, the relaxed decoder samples a query set $S \sim \mu_i$, obtains $y_S$ and outputs some deterministic function of $S$ and $y_S$.
        We can assume that this function is deterministic since the decoder is non-adaptive and never errs on valid codewords.
        Let $f_i: \D_i \times \Sigma^q \rightarrow \{0, 1\}$ be this deterministic function.

        \begin{algorithm} \caption{Global Sampler $\G$} \label{algorithm:global_decoder}
            \begin{algorithmic}[1]
                \Require{$\D_i, f_i, K_i \;\forall\; i \in[k]$, query access to $y \in \Sigma^n$.}
                \Ensure{$\hat{x} \in \{0, 1\}^k$}
                \State Sample $W \sim \mathrm{Bin}([n], p).$\;
                \If{$|W| \geq 2pn$}
                    \State Output a random $\hat{x} \sim \{0,1\}^k$.
                \EndIf
                \State $w \gets y_W$. 
                \For{$i \in [k]$}
                    \State $\D \gets \D_i, K \gets K_i, f \gets f_i, \hat{x}_i \gets 0$.
                    \State $\D^{\mathrm{sampled}} \gets \{S \in \D \mid S \setminus K \subseteq W\}$.  
                    \For{$\kappa \in \Sigma^{|K|}$}
                        \For{$S \in \D^ {\mathrm{sampled}}$}
                            \State $a_{S, \kappa} \gets (\kappa_{S \cap K},  w_{S \setminus K})$
                        \EndFor
                    \If{$\exists\; b \in \{0,1\}, \;\forall \;S \in \D^{\mathrm{sampled}}, f(S, a_{S, \kappa}) = b$}
                        \State $\hat{x}_i \gets b$. 
                    \EndIf
                    \EndFor
                \EndFor \\
            \Return $\hat{x}$.
            \end{algorithmic} 
    
        \end{algorithm}
    
    \paragraph{The global sampler.}
    The global sampler $\G$ operates in two phases: 
    \begin{enumerate}
        \item Query Phase: First $\G$ samples each coordinate of $y$ independently with probability $p$.
        Formally, it samples the coordinates $W \sim \text{Bin}([n], p)$.
        If $|W| \geq 2pn$, $\G$ outputs a random $\hat{x} \in \{0, 1\}^k$.
        
        Otherwise, $\G$ queries all the bits in $W$ from $y$ to obtain $w = y_W$. 
        \item Decoding Phase: For each $i \in [k]$, $\G$ decodes $x_i$ as follows.
        Let 
        \[
            \D_i^{\mathrm{sampled}} = \{S \in \D_i \mid S \setminus K_i \subseteq W\}
        \]
        be the query sets of $\D_i$ whose petals are fully sampled by $\G$.
        For any $\kappa \in \Sigma^{|K_i|}$ and $S \in \D_i^{\mathrm{sampled}}$, let $a_{S, \kappa}$ be the assignment of the variables of $S$ that is consistent with $\kappa$ on $K_i$ and with $w$ on $S \setminus K_i$. 
        
        $\G$ iterates over each assignment $\kappa \in \Sigma^{|K_i|}$ and does the following check on the relaxed decoder $\B$ with decoding function $f_i$: 

        For all $S \in \D_i^{\mathrm{sampled}}$ with corresponding assignment $a_{S, \kappa}$, does $\B$ output the same bit $b \in \{0, 1\}$?
        If yes, then $\G$ sets $\hat{x}_i = b$.
    
        If $\G$ never sets $x_i$ in these iterations, i.e., no such $\kappa$ exists, then it sets $\hat{x}_i = 0$. 
    \end{enumerate}
    
    \paragraph{Analysis.}
    We show that $\G$ succeeds, i.e., it outputs $x$ correctly, with probability at least $1/2$.
    $\G$ fails if it either samples $\geq 2pn$ coordinates, or decodes $x_i$ incorrectly for some $i \in [k]$.
    
    The first kind of failure happens if $|W| \geq 2pn$ when $W \sim \mathrm{Bin}([n], p)$. 
    By Chernoff's bound, $\Pr[|W| \geq 2pn] \leq \exp(-pn/3) \leq 1/n$, since $p \geq 3 \ln(n)/n$ by assumption.
     
    In \cref{claim:decoder_phase}, we argue that for each $i \in [k]$, $\G$ fails to decode $x_i$ with probability at most $|\Sigma|^{|K_i|}\cdot\varepsilon_i^{\sigma}$.
    Since, by the hypothesis, $\varepsilon_i < \exp \left(-\frac{\log(3k)+|K_i| \log(|\Sigma|)}{\sigma}\right)$, this failure probability is at most $1/3k$.
    
    First, we complete our argument assuming this claim:
    \begin{align*}
        \Pr_{W}[\G \text{ fails}] &= \Pr_{W}[\G \text{ fails and } |W| \geq 2pn] + \Pr_{W}[\G \text{ fails and }|W| \leq 2pn] \\ 
        & \leq \Pr_{W}[|W| \geq 2pn] + \Pr_{W}[\; \exists\; i \in [k], \G \text{ fails to decode $x_i$}] \\
        & \leq  1/n + \sum_{i \in [k]} \Pr_W[\G \text{ fails to decode $x_i$}] \\ 
        &\leq 1/n + k \cdot 1/3k \leq 1/2 \;,
    \end{align*}
    Hence, $\G$ decodes all of $x$ with probability at least $1/2$.
    Since it makes at most $2pn$ queries, this implies $k < 2pn \log|\Sigma|$.
    
    We are left to prove \cref{claim:decoder_phase}, which ensures that $\G$ decodes each $x_i$ with high probability, which finishes the proof.
    
    \begin{claim}\label{claim:decoder_phase}
        For any $i \in [k]$, the global sampler decodes $x_i$ with probability at least $1 - |\Sigma|^{|K_i|}\varepsilon_i^{\sigma}$. 
    \end{claim}
    
    \begin{proof}
        Fix $i \in [k]$.
        Let $\A := \B_i$ be the relaxed decoder with query set distribution $\mu := \mu_i$ which is a $(p, \varepsilon)$-robust daisy with kernel $K := K_i$ for $\varepsilon := \varepsilon_i$.
        Let $\D := \D_i$ and $f := f_i$. 

        Recall that $\G$ iterates over each $\kappa \in \Sigma^{|K|}$.
        In the case where $\G$ guesses $\kappa$ correctly, i.e., $\kappa = y_K$, we argue that $\G$ sets $\hat{x}_i$ correctly with probability at least $1-\varepsilon$.
        For every other guess, we argue that with probability at least $1-\varepsilon^{\sigma}$, the global sampler either sets $\hat{x}_i$ correctly or does not change $\hat{x}_i$ .
        By the union bound we then conclude that $\G$ sets $\hat{x}_i$ correctly with probability at least $1-|\Sigma|^{|K|}\varepsilon^{\sigma}$. 

        \paragraph{Correct guess.} 
        Suppose $\kappa = y_K$.
        We next show that $\G$ sets $x_i$ correctly with probability at least $1-\varepsilon$. 
        For each $S \in \D^{\mathrm{sampled}}$, $a_{S, \kappa}$ agrees with the codeword $y$.
        Since the relaxed decoder never errs on valid codewords, $f(S, a_{S, \kappa}) = x_i$ for all $S \in  \D^{\mathrm{sampled}}$.
        Therefore, as long as $|\D^{\mathrm{sampled}}| \geq 1$, $\G$ sets $x_i$ correctly.
        Since $\mu$ is a $(p, \varepsilon)$-robust daisy and $\D = \supp(\mu)$, this occurs with probability at least $1 - \varepsilon$. 

        \paragraph{Incorrect guess.}
        Conversely, suppose that $\kappa \neq y_K$.
        We show that with probability at least $1-\varepsilon^{\sigma}$, $\G$ does not set $x_i$ incorrectly.
        
        Let $y' \in \{0, 1\}^{n}$ be the string that agrees with the guess $\kappa$ on $K$, and agrees with $y$ on $[n] \setminus K$.
        Now, each $a_{S, \kappa}$ is consistent with $y'$, and $\dist(y, y') \leq \frac{|K|}{n} \leq \delta$.
        Therefore, due to the soundness of the decoder:
        \begin{equation*}
            \Pr_{S \sim \mu}[f(S, a_{S, \kappa}) \in \{x_i, \bot\}] \geq \sigma \; .
        \end{equation*}

        For any $\kappa \in \{0,1\}^K$, let $\D^{\text{good}}(\kappa) \subseteq \D$ be the query sets on which $\A$ decodes correctly when $K$ is assigned $\kappa$, i.e.,
        $\D^{\text{good}}(\kappa) = \left\{S \in \D \mid f(S, a_{S, \kappa}) \in \{x_i, \bot\}\right\}$, and note that the equation above implies $\mu(\D^{\text{good}}(\kappa)) \geq \sigma$.

        Since $\mu$ is a $(p, \varepsilon)$-robust daisy and $\D^{\text{good}}(\kappa) \subseteq \D$, $\D^{\mathrm{sampled}}$ contains a set from $\D^{\text{good}}(\kappa)$ with probability at least $1 - \varepsilon^{\sigma}$.
        Hence, in this case, there exists $ S \in \D^{\text{good}}(\kappa) \cap \D^{\mathrm{sampled}}$, which ensures that $\G$ does not set $x_i$ incorrectly.
        
    \end{proof}
\end{proof}

\subsection{Lower bound for arbitrary RLDCs} \label{section:lb-general}
    In this section, we use the robust daisy lemma (\pref{lemma:robust_daisy}) to show that an arbitrary RLDC can be transformed into an RLDC with a decoder whose query set distribution is a robust daisy with slightly worse soundness error.

\begin{lemma} \label{lemma:ldc_to_robust}
   Let $C :\{0,1\}^k\rightarrow \Sigma^n$ be an error correcting code with a non-adaptive $(q, \delta, \sigma)$-relaxed decoder.    
   
    Then, for every integer $c > \frac{q}{\sigma}$ such that $n^{-1/c} \leq \delta$, the code $C$ also has a non-adaptive $(q, \delta, \sigma')$-relaxed decoder where $\sigma' = \frac{c\sigma - q}{c-q}$, and for every $i \in [k]$, the query distribution $\mu_i$ of the new decoder for index $i$ is a $(p, \varepsilon_i)$-robust daisy with kernel $K_i \subseteq [n]$ with
    \[
        p =  \frac{8q^3 \cdot\log 3k \cdot \log n \cdot \log|\Sigma|}{\sigma-q/c} \cdot n^{-1/c}
    \]
    such that
    \begin{align*}
        |K_i| &\leq \delta n, & 
        \varepsilon_i^{\sigma'} &\leq \frac{1}{3k \cdot |\Sigma|^{|K_i|}} \: .
    \end{align*}
\end{lemma}

\begin{proof}
    We begin by describing our modified relaxed decoder. 
    
    \paragraph{The modified decoder.} 
    For every $i \in [k]$, let $\B_i$ be the original $(q, \delta,\sigma)$-relaxed decoder for $C$ for index $i$, with query set distribution $\mu_i$.
    
    Let $c > q/\sigma$ be an integer such that $n^{-1/c} \leq \delta$, and let 
    \[
        \alpha := \frac{8 q^3 \cdot \log 3k \cdot \log n \cdot \log|\Sigma|}{\sigma - q/c} \: .
    \]
    By \cref{lemma:robust_daisy}, there exists a subset $\D_i \subseteq \supp(\mu_i)$ such that $\mu_{\D_i}$ is a $(p, \varepsilon_i)$-robust daisy with non-empty kernel $K_i$, where: 

    \begin{align} \label{eqn:robust_daisy}
        p &= \alpha n^{-1/c} & \varepsilon_i &= \exp \left(-\frac{\alpha(1-q/c)}{8 q^2\log|\D_i|} \cdot |K_i|\right) & |K_i| &\leq n^{1-1/c} & \mu(\D_i) &\geq 1 - q/c \;. 
    \end{align}

    The modified decoder $\A$ works as follows.
    For each $i \in [k]$, the decoder for index $i$ samples a set $S \sim \mu_{\D_i}$.
    It then decodes $x_i$ according to the deterministic function $f_i$ of $\B_i$ restricted to $\D_i$.

    Since $\D_i \subseteq \supp(\mu_i)$, every set it samples is also a set the original decoder could have sampled.
    Therefore, the new decoder also makes at most $q$ queries, and never errs on valid codewords.

    \paragraph{Soundness.}
    Recall that $\sigma' := \frac{c\sigma - q}{c-q}$.
    We next show that the soundness probability of $\A$ is at least $\sigma'$.
    
    That is, we need to  show that for any $x \in \{0,1\}^k$ and $w \in \Sigma^n$ such that $\dist(w, C(x)) \leq \delta$, we have
    \begin{align} \label{eq:soundness}
        \Pr_{S \sim \mu_{\D_i}}[\A^w(i) =x_i] \geq \frac{c \sigma - q}{c - q}.
    \end{align}  
    Now,

    \begin{align*}
        1 - \sigma & \geq \Pr_{S \sim \mu_i}[\B^w(i) \neq x_i] \\
        &= \Pr_{S \sim \mu_{\D_i}} [\B^w(i) \neq x_i] \mu(\D_i) + \Pr_{S \sim \mu_{\overline{\D_i}}} [\B^w(i) \neq x_i] \mu(\overline{\D_i}) \\
        &\geq \Pr_{S \sim \mu_{\D_i}} [\B^w(i) \neq x_i] (1-q/c) = \Pr_{S \sim \mu_{\D_i}} [\A^w(i) \neq x_i] (1-q/c)
    \end{align*} 

    where $\overline{\D_i} = \supp(\mu_i) \setminus \D_i$ and using $\mu(\D_i) \geq 1- q/c$.
    Rearranging the inequality gives \cref{eq:soundness}. 
    
    We conclude that $\A$ is a $(q, \delta, \sigma')$-relaxed decoder for $C$. 

    \paragraph{Parameter verification.}
    By \cref{eqn:robust_daisy}, for each $i \in [k]$ the query set distribution $\mu_{\D_i}$ is a $(p, \varepsilon_i)$-robust daisy with kernel $K_i$ where: 
    \[
        p = \alpha n^{-1/c} = \frac{8q^3 \log n \log 3k \log|\Sigma|}{\sigma - q/c} \cdot n^{-1/c} \; .
    \]
    By assumptions, $|K_i| \leq n^{1-1/c} \leq \delta n$.
    
    Since each set in $\D_i$ has at most $q$ elements, we have $|\D_i| \leq n^{q} \implies q \log n \geq \log |\D_i|$.
    Plugging the definition of $\sigma'$ in the choice of $\alpha$, we infer the following bound:
    \[
        \alpha = \frac{8q^3 \log n \log 3k \log |\Sigma|}{\sigma - q/c} \geq \frac{8q^2\log |\D_i|}{\sigma'(1-q/c)} \cdot \log (3k) \cdot \log|\Sigma|,
    \]
    and then:
    \[
        \varepsilon_i^{\sigma'} = \exp \left(-\frac{\alpha \cdot (1-q/c)}{8q^2 \log |\D_i|}   \cdot |K_i| \cdot \sigma'\right) \leq \exp \left(-\log(3k)\cdot|K_i| \cdot \log |\Sigma| \right)\leq \frac{1}{3k \cdot |\Sigma|^{|K_i|}},
    \] 
    as required.
\end{proof}

Finally, we combine \cref{lemma:ldc_to_robust} and \cref{lemma:daisy-to-lowerbound} to derive our main theorem.
\begin{proof} [Proof of \pref{thm:main}]
    We prove that:
    \begin{align} \label{eqn:main}
        \frac{k}{\log^2 k} \leq 38q^4\sigma^{-2} \cdot \log^2 |\Sigma| \cdot n^{1-\frac{1}{ \left\lceil \frac{q}{\sigma} \right\rceil+1}} \; .
   \end{align}
   This yields the statement of \pref{thm:main} after rearrangement of the terms.
   
    Let $c := \lceil \frac{q}{\sigma} \rceil+1$.
    Note that $q/\sigma < c < 2q/\sigma$, and therefore, $n^{-1/c} \leq n^{-\frac{\sigma}{2q}} \leq \delta$. 
    
    By \cref{lemma:ldc_to_robust}, $C$ has a non-adaptive $(q, \delta, \sigma')$-relaxed decoder $\B$ with $\sigma' = \frac{c\sigma -q}{c - q}$ where for every $i \in [k]$, the query distribution $\mu_i$ of $\B$ on input $i$ is a $(p, \varepsilon_i)$-robust daisy with a non-empty kernel $K_i \subseteq [n]$ that satisfies \cref{eqn:robust_condition} with $p$ set according to the statement of \cref{lemma:ldc_to_robust}. 
    
    Note that $\sigma - q/c = \sigma - \frac{q}{\lceil q/\sigma \rceil + 1} \geq \sigma - \frac{q}{ q/\sigma + 1} = \frac{\sigma^2}{\sigma + q} \geq \frac{\sigma^2}{q+1}$.
    Furthermore, we can assume without loss of generality that $k \geq n^{1-\frac{1}{q+1}}$, since otherwise we are already done.
    Therefore, $\log n \leq \frac{(q+1) \log k}{q} \leq 2 \log k$ (since $q \geq 2$).
    Combining these two inequalities, 

    \[
        \frac{8q^3 \cdot \log 3k \cdot \log n \cdot \log |\Sigma| }{\sigma - q/c} \leq 8q^3 \cdot (\log 3 + \log k) \cdot \frac{q+1}{\sigma^2} \cdot 2\log k \cdot \log |\Sigma| < 19 q^4 \sigma ^{-2} \log^2 k \log |\Sigma| .
    \]
    
    Now, since $\B$ satisfies the guarantees of \cref{eqn:robust_condition} with 
    \[
        p = \frac{8q^3 \cdot \log 3k \cdot \log n \cdot \log |\Sigma|}{\sigma-q/c} \cdot n^{-1/c} < 19q^4 \sigma ^{-2} \log^2 k \cdot \log |\Sigma| \cdot n^{-1/c}
    \]
    
   we can apply \cref{lemma:daisy-to-lowerbound} (and note that $p > \frac{3 \ln n}{n}$ by the choice of parameters) to conclude:
    \[
        k < 2pn \log |\Sigma| < 38 q^4\sigma^{-2} \log^2 k \cdot \log^2 |\Sigma| \cdot n^{1-\frac{1}{\left\lceil \frac{q}{\sigma} \right\rceil+1}} \; .
    \]
\end{proof}

By rearrangement of the terms in \cref{eqn:main}, we derive a query lower bound for RLDCs with constant rate.

\begin{corollary}[Constant rate] \label{corollary:non_adaptive_const_rate}
   Fix a constant $\sigma > 1/2$ and let $\delta > n^{-1/4q}$.
   Let $C :\{0,1\}^k\rightarrow \{0,1\}^n$ be a non-adaptive $(q, \delta,\sigma)$-RLDC, and suppose that $n = O(k)$.
   Then, $q = \Omega\left(\frac{\log k}{\log \log k}\right)$.
\end{corollary}

\begin{proof}
    Suppose $n < d \cdot k$ for some constant $d > 0$.
    Note that $\lceil q/\sigma \rceil + 1 < 2q+2$. By \pref{thm:main}, we know that
    \[
        \frac{k}{\log^2 k} \leq 38q^4 \sigma^{-2}\cdot n^{1-\frac{1}{ \lceil \frac{q}{\sigma} \rceil+1}} < q^5\cdot d \cdot k^{1-\frac{1}{2q+2}} < q^6 \cdot k^{1-\frac{1}{2q+2}}
    \]
    where we use that $q = \omega(1)$ (otherwise it is easy to verify that $n = \omega(k)$) and $\sigma > 1/2$.

    This implies that    
    \[
        k^{\frac{1}{2q+3}} \leq \frac{k^{\frac{1}{2q+2}}}{\log^2 k} \leq q^6 
    \]
    
    where the first inequality follows for a large enough $k$.
    We conclude that
    \[
        \log k < 6(2q+3) \log q,
    \]
    which directly implies the stated bound. 
\end{proof}

Our lower bound for non-adaptive RLDCs extends to the important case of linear RLDCs, by using a known reduction. 

\begin{remark} \label{remark:linear_ldc}
    \cite{Goldberg24} showed that any linear $(q, \delta, \sigma)$-RLDC can be turned into $(q+1, \delta, \sigma)$-query non-adaptive RLDCs.
    We use this transformation along with \cref{thm:main} and \cref{corollary:non_adaptive_const_rate} to obtain \cref{corollary:linear_rldc} and  \cref{corollary:const-rate}. 
\end{remark} 

\section*{Acknowledgments}
Tom Gur thanks Victor Seixas Souza and Marcel Dall’Agnol for insightful conversations about $t$-daisies.
Sidhant Saraogi thanks Alexander Golovnev and Chao Yan for helpful discussions.
Guy Goldberg thanks Irit Dinur, Oded Goldreich, and Guy Rothblum for valuable discussions about relaxed LDCs.

\bibliographystyle{alpha}
\bibliography{references}

@article{gl19,
author = {Gur, Tom and Lachish, Oded},
title = {On the Power of Relaxed Local Decoding Algorithms},
journal = {SIAM Journal on Computing},
volume = {50},
number = {2},
pages = {788-813},
year = {2021},
doi = {10.1137/19M1307834},
}

@article{ALWZ21,
author = {Ryan Alweiss and Shachar Lovett and Kewen Wu and Jiapeng Zhang},
title = {{Improved bounds for the sunflower lemma}},
volume = {194},
journal = {Annals of Mathematics},
number = {3},
publisher = {Department of Mathematics of Princeton University},
pages = {795 -- 815},
year = {2021},
doi = {10.4007/annals.2021.194.3.5},
URL = {https://doi.org/10.4007/annals.2021.194.3.5}
}

@article{GR18,
  title={Non-interactive proofs of proximity},
  author={Gur, Tom and Rothblum, Ron D},
  journal={Computational Complexity},
  volume={27},
  number={1},
  pages={99--207},
  year={2018},
  publisher={Birkhauser Verlag Basel, Switzerland, Switzerland}
}

@article{GG18,
  title={Universal Locally Testable Codes},
  author={Goldreich, Oded and Gur, Tom},
  journal={Chicago Journal OF Theoretical Computer Science},
  volume={3},
  pages={1--21},
  year={2018}
}

@article{GG21,
  title={Universal locally verifiable codes and 3-round interactive proofs of proximity for CSP},
  author={Goldreich, Oded and Gur, Tom},
  journal={Theoretical computer science},
  volume={878},
  pages={83--101},
  year={2021},
  publisher={Elsevier}
}

@article{DGMT22,
  title={Quantum proofs of proximity},
  author={Dall'Agnol, Marcel and Gur, Tom and Moulik, Subhayan Roy and Thaler, Justin},
  journal={Quantum},
  volume={6},
  pages={834},
  year={2022},
  publisher={Verein zur F{\"o}rderung des Open Access Publizierens in den Quantenwissenschaften}
}

@article{BCW21,
title = {Note on sunflowers},
journal = {Discrete Mathematics},
volume = {344},
number = {7},
pages = {112367},
year = {2021},
issn = {0012-365X},
doi = {https://doi.org/10.1016/j.disc.2021.112367},
url = {https://www.sciencedirect.com/science/article/pii/S0012365X21000807},
author = {Tolson Bell and Suchakree Chueluecha and Lutz Warnke},
}

@article{DGL23,
  title={A Structural Theorem for Local Algorithms with Applications to Coding, Testing, and Verification},
  author={Dall’Agnol, Marcel and Gur, Tom and Lachish, Oded},
  journal={SIAM Journal on Computing},
  volume={52},
  number={6},
  pages={1413--1463},
  year={2023},
  publisher={SIAM}
}

@article{BGHSV06,
author = { Ben‐Sasson, Eli and  Goldreich, Oded and  Harsha, Prahladh and  Sudan, Madhu and  Vadhan, Salil},
title = {Robust PCPs of Proximity, Shorter PCPs, and Applications to Coding},
journal = {SIAM Journal on Computing},
volume = {36},
number = {4},
pages = {889-974},
year = {2006},
doi = {10.1137/S0097539705446810},

URL = { 
    
        https://doi.org/10.1137/S0097539705446810
    
    

},
}

@article{Goldreich24,
  author       = {Oded Goldreich},
  title        = {On the relaxed {LDC} of {BGHSV:} {A} survey that corrects the record},
  journal      = {Electron. Colloquium Comput. Complex.},
  volume       = {{TR24-078}},
  year         = {2024},
  url          = {https://eccc.weizmann.ac.il/report/2024/078},
  eprinttype    = {ECCC},
  eprint       = {TR24-078},
  bibsource    = {dblp computer science bibliography, https://dblp.org}
}

@article{Goldreich23,
  author       = {Oded Goldreich},
  title        = {On the Lower Bound on the Length of Relaxed Locally Decodable Codes},
  journal      = {Electron. Colloquium Comput. Complex.},
  volume       = {{TR23-064}},
  year         = {2023},
  url          = {https://eccc.weizmann.ac.il/report/2023/064},
  eprinttype    = {ECCC},
  eprint       = {TR23-064},
  bibsource    = {dblp computer science bibliography, https://dblp.org}
}

@InProceedings{Goldberg24,
  author =	{Goldberg, Guy},
  title =	{{Linear Relaxed Locally Decodable and Correctable Codes Do Not Need Adaptivity and Two-Sided Error}},
  booktitle =	{51st International Colloquium on Automata, Languages, and Programming (ICALP 2024)},
  pages =	{74:1--74:20},
  series =	{Leibniz International Proceedings in Informatics (LIPIcs)},
  ISBN =	{978-3-95977-322-5},
  ISSN =	{1868-8969},
  year =	{2024},
  volume =	{297},
  editor =	{Bringmann, Karl and Grohe, Martin and Puppis, Gabriele and Svensson, Ola},
  publisher =	{Schloss Dagstuhl -- Leibniz-Zentrum f{\"u}r Informatik},
  address =	{Dagstuhl, Germany},
  URL =		{https://drops.dagstuhl.de/entities/document/10.4230/LIPIcs.ICALP.2024.74},
  URN =		{urn:nbn:de:0030-drops-202174},
  doi =		{10.4230/LIPIcs.ICALP.2024.74}
}

@inproceedings{KM24,
author = {Kumar, Vinayak M. and Mon, Geoffrey},
title = {Relaxed Local Correctability from Local Testing},
year = {2024},
isbn = {9798400703836},
publisher = {Association for Computing Machinery},
address = {New York, NY, USA},
url = {https://doi.org/10.1145/3618260.3649611},
doi = {10.1145/3618260.3649611},
booktitle = {Proceedings of the 56th Annual ACM Symposium on Theory of Computing},
pages = {1585–1593},
numpages = {9},
location = {Vancouver, BC, Canada},
series = {STOC 2024}
}

@InProceedings{CY24,
  author =	{Cohen, Gil and Yankovitz, Tal},
  title =	{{Asymptotically-Good RLCCs with $\log^{2+o(1)} n$ Queries}},
  booktitle =	{39th Computational Complexity Conference (CCC 2024)},
  pages =	{8:1--8:16},
  series =	{Leibniz International Proceedings in Informatics (LIPIcs)},
  ISBN =	{978-3-95977-331-7},
  ISSN =	{1868-8969},
  year =	{2024},
  volume =	{300},
  editor =	{Santhanam, Rahul},
  publisher =	{Schloss Dagstuhl -- Leibniz-Zentrum f{\"u}r Informatik},
  address =	{Dagstuhl, Germany},
  URL =		{https://drops.dagstuhl.de/entities/document/10.4230/LIPIcs.CCC.2024.8},
  URN =		{urn:nbn:de:0030-drops-204045},
  doi =		{10.4230/LIPIcs.CCC.2024.8}
}

@INPROCEEDINGS{CY22,
  author={Cohen, Gil and Yankovitz, Tal},
  booktitle={2022 IEEE 63rd Annual Symposium on Foundations of Computer Science (FOCS)}, 
  title={Relaxed Locally Decodable and Correctable Codes: Beyond Tensoring}, 
  year={2022},
  volume={},
  number={},
  pages={24-35},
  doi={10.1109/FOCS54457.2022.00010}}

@article{CGS22,
author = {Chiesa, Alessandro and Gur, Tom and Shinkar, Igor},
title = {Relaxed Locally Correctable Codes with Nearly-Linear Block Length and Constant Query Complexity},
journal = {SIAM Journal on Computing},
volume = {51},
number = {6},
pages = {1839-1865},
year = {2022},
doi = {10.1137/20M135515X},
URL = {https://doi.org/10.1137/20M135515X},
eprint = {https://doi.org/10.1137/20M135515X}
}

@InProceedings{AS21,
  author =	{Asadi, Vahid R. and Shinkar, Igor},
  title =	{{Relaxed Locally Correctable Codes with Improved Parameters}},
  booktitle =	{48th International Colloquium on Automata, Languages, and Programming (ICALP 2021)},
  pages =	{18:1--18:12},
  series =	{Leibniz International Proceedings in Informatics (LIPIcs)},
  ISBN =	{978-3-95977-195-5},
  ISSN =	{1868-8969},
  year =	{2021},
  volume =	{198},
  editor =	{Bansal, Nikhil and Merelli, Emanuela and Worrell, James},
  publisher =	{Schloss Dagstuhl -- Leibniz-Zentrum f{\"u}r Informatik},
  address =	{Dagstuhl, Germany},
  URL =		{https://drops.dagstuhl.de/entities/document/10.4230/LIPIcs.ICALP.2021.18},
  URN =		{urn:nbn:de:0030-drops-140878},
  doi =		{10.4230/LIPIcs.ICALP.2021.18}
}

@book{AS16,
author = {Alon, Noga and Spencer, Joel H.},
title = {The Probabilistic Method},
year = {2016},
isbn = {1119061954},
publisher = {Wiley Publishing},
edition = {4th},
}

@article{yekhanin08,
author = {Yekhanin, Sergey},
title = {Towards 3-query locally decodable codes of subexponential length},
year = {2008},
issue_date = {February 2008},
publisher = {Association for Computing Machinery},
address = {New York, NY, USA},
volume = {55},
number = {1},
issn = {0004-5411},
url = {https://doi.org/10.1145/1326554.1326555},
doi = {10.1145/1326554.1326555},
journal = {J. ACM},
month = feb,
articleno = {1},
numpages = {16}
}

@article{efr12,
author = {Efremenko, Klim},
title = {3-Query Locally Decodable Codes of Subexponential Length},
journal = {SIAM Journal on Computing},
volume = {41},
number = {6},
pages = {1694-1703},
year = {2012},
doi = {10.1137/090772721},
URL = {https://doi.org/10.1137/090772721},
eprint = {https://doi.org/10.1137/090772721}}

@article{GRR20,
 author = {Gur, Tom and Ramnarayan, Govind and Rothblum, Ron},
 title = {Relaxed Locally Correctable Codes},
 year = {2020},
 pages = {1--68},
 doi = {10.4086/toc.2020.v016a018},
 publisher = {Theory of Computing},
 journal = {Theory of Computing},
 volume = {16},
 number = {18},
 URL = {https://theoryofcomputing.org/articles/v016a018},
}

@article{GGK19,
author = {Goldreich, Oded and Gur, Tom and Komargodski, Ilan},
title = {Strong Locally Testable Codes with Relaxed Local Decoders},
year = {2019},
issue_date = {September 2019},
publisher = {Association for Computing Machinery},
address = {New York, NY, USA},
volume = {11},
number = {3},
issn = {1942-3454},
url = {https://doi.org/10.1145/3319907},
doi = {10.1145/3319907},
journal = {ACM Trans. Comput. Theory},
month = apr,
articleno = {17},
numpages = {38}
}

@ARTICLE{PP23,
  title     = {A proof of the {Kahn--Kalai} conjecture},
  author    = {Park, Jinyoung and Pham, Huy},
  journal   = {J. Amer. Math. Soc.},
  publisher = {American Mathematical Society (AMS)},
  volume    =  {37},
  number    =  {1},
  pages     = {235--243},
  month     =  {aug},
  year      =  {2023},
  language  = {en}
}

@inproceedings{KT00,
author = {Katz, Jonathan and Trevisan, Luca},
title = {On the efficiency of local decoding procedures for error-correcting codes},
year = {2000},
isbn = {1581131844},
publisher = {Association for Computing Machinery},
address = {New York, NY, USA},
url = {https://doi.org/10.1145/335305.335315},
doi = {10.1145/335305.335315},
booktitle = {Proceedings of the Thirty-Second Annual ACM Symposium on Theory of Computing},
pages = {80–86},
numpages = {7},
location = {Portland, Oregon, USA},
series = {STOC '00}
}

@article{rossman14,
author = {Rossman, Benjamin},
title = {The Monotone Complexity of \$k\$-Clique on Random Graphs},
journal = {SIAM Journal on Computing},
volume = {43},
number = {1},
pages = {256-279},
year = {2014},
doi = {10.1137/110839059},
URL = {https://doi.org/10.1137/110839059},
eprint = {https://doi.org/10.1137/110839059}
}

@article{Yekhanin12,
  author       = {Sergey Yekhanin},
  title        = {Locally Decodable Codes},
  journal      = {Found. Trends Theor. Comput. Sci.},
  volume       = {6},
  number       = {3},
  pages        = {139--255},
  year         = {2012}
}

@article{KS17,
  author       = {Swastik Kopparty and
                  Shubhangi Saraf},
  title        = {Local Testing and Decoding of High-Rate Error-Correcting Codes},
  journal      = {Electron. Colloquium Comput. Complex.},
  volume       = {{TR17-126}},
  year         = {2017}
}

@inproceedings{KW03,
  author       = {Iordanis Kerenidis and
                  Ronald de Wolf},
  title        = {Exponential lower bound for 2-query locally decodable codes via a
                  quantum argument},
  booktitle    = {{STOC}},
  pages        = {106--115},
  publisher    = {{ACM}},
  year         = {2003}
}

@article{Woodruff07,
  author       = {David P. Woodruff},
  title        = {New Lower Bounds for General Locally Decodable Codes},
  journal      = {Electron. Colloquium Comput. Complex.},
  volume       = {{TR07-006}},
  year         = {2007}
}

@article{Woodruff12,
  author       = {David P. Woodruff},
  title        = {A Quadratic Lower Bound for Three-Query Linear Locally Decodable Codes
                  over Any Field},
  journal      = {J. Comput. Sci. Technol.},
  volume       = {27},
  number       = {4},
  pages        = {678--686},
  year         = {2012}
}

@article{AGKM22,
  author       = {Omar Alrabiah and
                  Venkatesan Guruswami and
                  Pravesh Kothari and
                  Peter Manohar},
  title        = {A Near-Cubic Lower Bound for 3-Query Locally Decodable Codes from
                  Semirandom {CSP} Refutation},
  journal      = {Electron. Colloquium Comput. Complex.},
  volume       = {{TR22-101}},
  year         = {2022}
}

@article{Rao2020,
	author = {Rao, Anup},
	journal = {Discrete Analysis},
	doi = {10.19086/da.11887},
	year = {2020},
	month = {feb 25},
	title = {Coding for {Sunflowers}},
}

@inproceedings{talagrand10,
author = {Talagrand, Michel},
title = {Are many small sets explicitly small?},
year = {2010},
isbn = {9781450300506},
publisher = {Association for Computing Machinery},
address = {New York, NY, USA},
url = {https://doi.org/10.1145/1806689.1806693},
doi = {10.1145/1806689.1806693},
booktitle = {Proceedings of the Forty-Second ACM Symposium on Theory of Computing},
pages = {13–36},
numpages = {24},
location = {Cambridge, Massachusetts, USA},
series = {STOC '10}
}

@ARTICLE{MNSZ25,
  title     = "A Bayesian proof of the spread lemma",
  author    = "Mossel, Elchanan and Niles-Weed, Jonathan and Sun, Nike and
               Zadik, Ilias",            
  journal   = "Random Struct. Algorithms",
  publisher = "Wiley",
  volume    =  66,
  number    =  4,
  month     =  jul,
  year      =  2025,
  copyright = "http://creativecommons.org/licenses/by/4.0/",
  language  = "en"
}

@inproceedings{JM25,
  title={A $k^{q/(q-2)}$ Lower Bound for Odd Query Locally Decodable Codes from Bipartite {K}ikuchi Graphs},
  author={Janzer, Oliver and Manohar, Peter},
  booktitle={FOCS},
  year={2025}
}

@inproceedings{BHKL25,
  title={Improved Lower Bounds for all Odd-Query Locally Decodable Codes},
  author={Basu, Arpon and Hsieh, Jun-Ting and Kothari, Pravesh K. and Lin, Andrew D.},
  booktitle={FOCS},
  year={2025}
}

@inproceedings{LSZ20, author = {Lovett, Shachar and Solomon, Noam and Zhang, Jiapeng}, title = {From DNF compression to sunflower theorems via regularity}, year = {2020}, booktitle = {Proceedings of the 34th Computational Complexity Conference}, articleno = {5}, numpages = {14}, series = {CCC '19} }

@article{CKR22,
  title={Monotone circuit lower bounds from robust sunflowers},
  author={Cavalar, Bruno Pasqualotto and Kumar, Mrinal and Rossman, Benjamin},
  journal={Algorithmica},
  year={2022}
}

@article{CGRSS25,
  title={Monotone Circuit Complexity of Matching},
  author={Cavalar, Bruno and G{\"o}{\"o}s, Mika and Riazanov, Artur and Sofronova, Anastasia and Sokolov, Dmitry},
  journal ={arXiv:2507.16105},
  year={2025}
}

@ARTICLE{FKN21,
  title     = "Thresholds versus fractional expectation-thresholds",
  author    = "Frankston, Keith and Kahn, Jeff and Narayanan, Bhargav and Park,
               Jinyoung",
  journal   = "Ann. Math.",
  publisher = "Annals of Mathematics",
  volume    =  194,
  number    =  2,
  pages     = "475",
  month     =  sep,
  year      =  2021
}

@article{rao2025story,
  title={The Story of Sunflowers},
  author={Rao, Anup},
  journal={arXiv preprint arXiv:2509.14790},
  year={2025}
}
\end{document}